\documentclass[pra,twocolumn, aps,superscriptaddress]{revtex4-1}
\usepackage{amsmath,amsfonts,amssymb,amsthm,graphicx,bbm,enumerate,times,color}
\usepackage{mathtools}
\usepackage{soul}

\newtheorem{theorem}{Theorem}

\newtheorem{result}{Result}

\newtheorem{lemma}[theorem]{Lemma}

\newtheorem{example}[theorem]{Example}

\usepackage{hyperref}

\newcommand{\average}[1]{\left\langle #1\right\rangle}

\newcommand{\ii}{{\rm i}}
\newcommand{\e}{{\rm e}}

\newcommand{\id}{\mathbbm{1}}
\newcommand{\mc}[1]{\mathcal{#1}}
\newcommand{\bra}[1]{\langle #1|}
\newcommand{\ket}[1]{|#1\rangle}
\newcommand{\braket}[2]{\langle #1|#2\rangle}
\newcommand{\ketbra}[2]{| #1 \rangle \langle #2 |}
\newcommand{\expect}[1]{\langle #1\rangle}
\newcommand{\proj}[1]{\vert #1\rangle\!\langle#1 \vert}
\renewcommand{\d}{\mathrm{d}}

\newcommand{\Tr}{\operatorname{tr}}
\newcommand{\tr}{\Tr}

\newcommand{\ta}{{\Omega}_{\text{TA}}} 
\newcommand{\gibbs}{\Omega_{\text{Gibbs}}} 
\newcommand{\gge}{\Omega_{\text{GGE}}} 
\newcommand{\argmax}{\mathrm{argmax}}

\newcommand{\fu}{Dahlem Center for Complex Quantum Systems, Freie Universit{\"a}t Berlin, 14195 Berlin, Germany}

\begin{document}
\title{Work and entropy production in generalised Gibbs ensembles}

\author{Mart\'i Perarnau-Llobet}
\affiliation{ICFO-Institut de Ciencies Fotoniques, The Barcelona Institute of Science and Technology, 08860 Castelldefels (Barcelona), Spain}
\author{Arnau Riera}
\affiliation{ICFO-Institut de Ciencies Fotoniques, The Barcelona Institute of Science and Technology, 08860 Castelldefels (Barcelona), Spain}
\author{Rodrigo Gallego}
\affiliation{\fu}
\author{Henrik Wilming}
\affiliation{\fu}
\author{Jens Eisert}
\affiliation{\fu}

\date{\today}

\begin{abstract}
Recent years have seen an enormously revived interest in the study of thermodynamic notions in the quantum regime. This applies both to the study of notions of work extraction in thermal machines in the quantum regime, as well as to questions of equilibration and thermalisation of interacting quantum many-body systems as such. In this work we bring together these two lines of research by studying work extraction in a closed system that undergoes a sequence of quenches and equilibration steps concomitant with free evolutions. In this way, we incorporate an important insight from the study of the dynamics of quantum many body systems: the evolution of closed systems is expected to be well described, for relevant observables and most times, by a suitable equilibrium state. We will consider three kinds of equilibration, namely to (i) the time averaged state, (ii) the Gibbs ensemble and (iii) the generalised Gibbs ensemble (GGE), reflecting further constants of motion in integrable models. For each effective description, we investigate notions of entropy production, the validity of the minimal work principle and properties of optimal work extraction protocols. While we keep the discussion general, much room is dedicated to the discussion of paradigmatic non-interacting fermionic quantum many-body systems, for which we identify significant differences with respect to the role of the minimal work principle. Our work not only has implications for experiments with cold atoms, but also can be viewed as suggesting a mindset for quantum thermodynamics where the  role of the external heat baths is instead played by the system itself, with its internal degrees of freedom bringing  coarse-grained observables to equilibrium.
\end{abstract}

\maketitle

Thermodynamics is undoubtedly one of the most successful physical theories, accurately
describing a vast plethora of situations and phenomena. Until not too long ago, the study of thermodynamic
state transformations was mostly confined to the realm of classical physics, which constitutes
a most meaningful approach when considering macroscopic situations. Progress on the 
precisely controlled manipulation of physical systems at the nano-scale or at the level of single atoms,
however, has pushed the frontier of the applicability of thermodynamic notions to the realm of 
quantum physics. Indeed, the emergent research field of quantum thermodynamics is concerned with 
thermodynamics in the quantum regime, a regime 
in which notions of coherence,  strong interactions, and entanglement are 
expected to play a significant role. 

Building upon a body of early work \cite{GeuSchSco67,Alicki1979}, recent attempts of 
grasping the specifics emerging in the extreme quantum regime have put particular emphasis on 
notions of thermodynamic state transformations for quantum systems. A similar focus has been put on 
studying the  rates of achievable work extraction of thermodynamic machines
\cite{ThermoReview,DelRio2011,Hilt2011,Faist2012,Egloff2012,Horodecki2013a,Brandao2013b,Aberg2013,Reeb2013,StrongCoupling,Faist2014b,Frenzel2014,AndersEditorial,Wilming2014,Gelbwaser2015}. 
In these new attempts, a resource-theoretic mindset is often applied, or single-shot notions of work extraction 
\cite{Aberg2013a,Horodecki2011} are elaborated upon.
These studies are motivated by foundational considerations----after all, such 
thermodynamic state transformations are readily available in a number of quantum architectures---as 
well as by technological desiderata: For example, novel techniques for cooling quantum systems
close to the ground state can be derived from quantum thermodynamical considerations
\cite{Brunner2013,Correa2013a}. 
In these studies of quantum heat engines, heat baths prepared in thermal states are usually 
still taken for granted: This is most manifest in a resource-theoretic language, where such 
thermal baths in Gibbs states are considered a free resource.

Concomitant with these recent studies of thermal machines, a second branch of quantum thermodynamics 
is blossoming: This is the study of quantum many-body systems out of equilibrium and the question of
thermalisation as such \cite{CalabreseCardy06,CramerEisert,Rigol2008,Linden2009,Short2012,Reimann2008,CauxEssler}. 
In this context, thermal baths are by no means assumed to be available: Instead
it is one of the key tasks of this field of research to find out under what precise conditions 
closed many-body systems are expected to thermalise, following quenches out of equilibrium.
This is hence the question in what precise sense systems---as one often says----``form their own heat bath''. Despite 
respectable progress in recent years, many questions on many-body systems out of equilibrium 
remain open, even when it comes to understanding whether non-integrable
generic systems always thermalise at all \cite{Gogolin2015}. Many-body localised systems are expected to stubbornly refuse to 
thermalise, for retaining information of the initial condition over an infinite amount of time. Integrable models, in contrast, are not equilibrating 
to Gibbs states, but to so-called generalised Gibbs ensembles 
(GGE) \cite{CramerEisert,BrodyGGE,GeneralisedGibbs,CauxEssler,0906.1663,1302.6944,1205.2211,ProsenGGE,CompleteGGE}. 
For comprehensive reviews on the subject,
see, e.g., Refs.\ \cite{1408.5148,PolkovnikovReview,Gogolin2015,AnatoliMarcosReview}.

It is the purpose of this work to bring these two realms of study closer together and to attempt to formulate a theory of quantum thermodynamics and notions of work extraction, taking into account these recent insights into the mechanism of equilibration in many-body systems. More specifically, we consider work extraction from a closed system that undergoes a sequence of quenches and relaxations to a
respective equilibrium state. Importantly, our framework deviates from the standard realm of thermodynamics, where equilibration to statistical ensembles after each quench occurs through weak coupling with an infinite thermal bath. In contrast, we incorporate the equilibration to such ensembles as an effective description of the unitary evolution of a closed system. This effective description is adequate to capture the system only for a restricted, although most relevant, set of observables. We will consider three kinds of equilibrium states: the time averaged state, the Gibbs ensemble, and the generalised Gibbs ensemble for a given set of constants of motion.   Entropy production and the minimal work principle will be studied for these three models.

The results presented here are expected to be of interest for both the study of thermal machines in the quantum regime---since new 
 insights for the equilibration of closed quantum many-body is taken into account---as well as for the study of 
 quantum many-body equilibration itself. Our work highlights the importance of investigating not only the equilibration of systems after single quenches, but also the equilibration after sequences of quenches which are the relevant paradigm within protocols of work extraction.

The structure of this work is as follows. {In Sec.\ \ref{sec:equilibrationmodels} we introduce the three models of equilibration that will be considered throughout this work and discuss its physical relevance as a description of the effective evolution of closed many-body systems. In Sec.\ \ref{Sec:OperationsAndEquilibration} we turn to presenting our framework of work extraction based on quenches and equilibrations. Sec.\ \ref{sec:entropyproduction} discusses} notions of entropy production in each of the models of equilibration, where we introduce rigorous conditions for the absence of entropy production and carefully relate
these conditions to notions of reversible processes. In Sec.\ \ref{sec:workextraction} we discuss the minimal work principle and the protocols for optimal work extraction for each of the models of equilibration. Lastly, in Sec.\ \ref{sec:freefermions} we study a model of non-interacting fermionic systems, where many of the features throughout our theoretical analysis are made concrete.

\section{Equilibration models}\label{sec:equilibrationmodels}
When referring to equilibration of quantum many-body systems, we relate to finite but large systems.
Such closed quantum many-body systems cannot truly equilibrate due to their unitary evolution. 
What is generically the case, however, is that expectation values  of large restricted sets of 
observables  equilibrate in time to the value attained for the time average
\cite{Reimann2008,Reimann2012,Short2011,CramerEisert}, in the sense that they stay close to the time average for most times
in an overwhelming majority. This is particularly true for local observables \cite{GarciaPintos2015}.


\subsection{Time average state or diagonal ensemble}
We say that an observable $A$ equilibrates if, after some relaxation time, its expectation value is for most 
times the same  $\expect{A(t)}\simeq \Tr(A\ta)$ as the expectation value of the \emph{infinite time average}
\begin{equation}
\ta(\rho,H) := \lim_{T\rightarrow \infty}\frac{1}{T}\int_{0}^T \e^{-\ii Ht}\,\rho\, \e^{\ii Ht}\mathrm{d}t\, ,
\label{eq:time-average}
\end{equation}
of an initial state $\rho$ of a system described by a Hamiltonian $H$.
A simple calculation shows that the time averaged state 
corresponds to the de-phased state in the Hamiltonian eigenbasis and for this reason is often called \emph{diagonal ensemble}.
More explicitly, given the distinct energies of the Hamiltonian $\{E_k\}$ 
and the projectors onto their corresponding eigenspaces $P_k$, 
the time averaged state reads
\begin{equation}
\ta(\rho,H) = \sum_k P_k\rho P_k\, .
\end{equation}
The time averaged state corresponds to the maximum entropy state given all the conserved quantities  \cite{Gogolin2011}. 
This observation turns the \emph{principle of maximum entropy} introduced by Jaynes \cite{Jaynes1957,Jaynes1957a} into a consequence of the quantum dynamics. The principle of maximum entropy states that the probability distribution which best represents the current state of knowledge of the system is the one with largest entropy given the conserved quantities of the system; this principle will be crucial to define our equilibration models.

Although relaxation towards the time averaged state has been proven under very general and naturally fulfilled conditions \cite{Reimann2008,Linden2009,Short2012,Reimann2012}, in practice, the diagonal ensemble cannot be used as an equilibration model due to its inefficiency. 
The description of the equilibrium state by the diagonal ensemble requires the specification of as many conserved quantities as the dimension of the Hilbert space, which scales exponentially in the system size. It is therefore in principle not even possible to save all the data in a computer for a large interacting many-body system, let alone compute the infinite time average efficiently.

\subsection{Canonical or Gibbs ensemble}
In practice, the characterisation of the equilibrium state can in many instances be done by specifying only a few 
quantities, e.g., the temperature and the chemical potential. The most relevant and common such situation 
is the \emph{canonical ensemble} or the \emph{Gibbs state}, for which only the temperature,
or equivalently the energy per particle of the initial state $\rho$, has to be specified,
\begin{equation}\label{eq:Gibbs-def}
\Omega_{\text{Gibbs}}(\rho,H)=\frac{\e^{-\beta H}}{Z}\, ,
\end{equation}
where {$\rho$ is the state of the system before undergoing the equilibration process,} $Z=\Tr(\e^{-\beta H})$ is the partition function and 
the inverse temperature $\beta>0$ is fixed by imposing that $\Tr(H\gibbs)=\Tr(H\rho)$. 

For generic, non-integrable models, the thermal state is expected to be indistinguishable from the time averaged state
under very mild assumptions which relate to conditions on eigenstates of the Hamiltonian \cite{Srednicki1994,Rigol2008,Gogolin2015}
and on the energy distribution of the initial state \cite{Riera2012,Brandao2015b}. While dynamical
thermalisation in this sense has not yet been rigorously proven, it is highly plausible, and it can be connected to 
typicality arguments \cite{Popescu2006,Goldstein2006}. 
The generality of these conditions explains why the canonical ensemble 
is the corner-stone of the standard thermodynamics.
Nevertheless, there are known instances of systems that do not thermalise.
One central \emph{aim of this work is to study how thermodynamic protocols
 are modified when the Gibbs ensemble 
is not a good equilibration model and does not satisfactorily describe the equilibrium state of the system}.

\subsection{Generalised Gibbs ensemble}
Examples of systems which do not fully thermalise to Gibbs states are constituted by \emph{integrable systems}.
The infinite-time averaged states are not well described by the Gibbs ensemble because of the existence of (quasi) local integrals of motion, i.e. 
conserved quantities 
$Q_i$, that retain information about the initial state over an infinite amount of time. 
Instead, there is strong evidence that they can be well-described by the so-called generalised Gibbs ensemble (GGE)
defined as
\begin{equation}\label{eq:defGGE}
\gge(\rho,H,\{Q_i\})\propto 
\e^{-\beta H + \sum_{j=1}^{q} \lambda_j Q_j }\, ,
\end{equation}
where the generalised chemical potential $\lambda_j$ is a Lagrange multiplier 
associated with the specific conserved quantity $Q_j$, $j=1,\dots, q$, such that its expectation value is the same as 
the one of the initial state
\begin{equation}
\Tr\left(\gge (\rho,H,\{Q_j\}) \ Q_k \right)=\Tr(\rho Q_k)\, ,
\end{equation}
for each $k=1,\dots, q$.
The GGE can be understood as an interpolation between the diagonal and the canonical ensembles.
The diagonal ensemble maximises the von Neumann entropy $S(\rho)=-\Tr(\rho \log \rho)$
given all the conserved quantities (CQ).
The Gibbs ensemble maximises the von Neumann entropy considering only the
energy as a conserved quantity.
The GGE is situated in between. For a given state $\rho$ and a set of operators (conserved quantities)
$\{Q_i\}$, it is natural to define the set of states compatible with the values the conserved quantities 
\begin{equation}
\mathcal{E}(\rho,\{Q_i\}) := \{\sigma | \tr(\rho Q_i) = \tr(\sigma Q_i)\}.
\end{equation}
The GGE is the state that maximises the von Neumann entropy within $\mathcal{E}(\rho,\{Q_i\})$.
From this perspective, the ensembles introduced so far can be summarised as
\begin{align}\label{modelseq}
\ta&:=\argmax_{\sigma\in\mathcal{E}(\rho,\{\textrm{all CQ}\})} S(\sigma)\, ,   \\
\label{modelseq2}\gge(\rho,H,\{Q_i\}) &:= \argmax_{\sigma\in\mathcal{E}(\rho,\{H,Q_i\})} S(\sigma)\, ,  \\
\label{modelseq3}\gibbs(\rho,H)&:=\argmax_{\sigma\in\mathcal{E}(\rho,\{H\})} S(\sigma) \, .
\end{align}

A relevant question in the construction of GGEs is how the conserved quantities have to be chosen, which is discussed in Appendix \ref{sec:appGGE}. 
In general, there is a certain degree of ambiguity of what constants of motion to pick in order to arrive at the appropriate
GGE. This discussion is not relevant for the general study pursued in this work, however.
It is the aim of this work to study the thermodynamical behaviour of the GGE in full generality, 
hence we will not have to make any precise assumption about the conserved quantities, unless it is explicitly specified.

\subsection{Example: Equilibration of a quadratic fermionic model}
To illustrate the above considerations, let us consider a quadratic Hamiltonian of fermions in a one dimensional lattice
\begin{equation}\label{eq:Hfermchain0}
H^{(0)}= \sum_{i=1}^{n}  \epsilon_i a_i^{\dagger} a_i +  g \sum_{i=1}^{n-1}  \left(a_i^{\dagger} a_{i+1}+a_{i+1}^{\dagger} a_{i}\right),
\end{equation}
where $n$ is the total number of sites and  $a_i$ ($a_i^{\dagger}$) are the creation (annihilation) operators at the 
$i$-site which satisfy the fermionic 
anti-commutation relations 
\begin{equation}
\{a_i,a_j^{\dagger}\}=\delta_{i,j},\,
\{a_i,a_j\}=\{a_i^{\dagger},a_j^{\dagger}\}=0.
\end{equation}
We would like to study how an initially out of equilibrium state relaxes to equilibrium
and see that the Gibbs ensemble fails to describe the equilibrium state.

The initial state of the system is taken to be in thermal equilibrium,
$\rho^{(0)} = e^{-\beta H^{(0)}}/\mathcal{Z}$. A quench is then performed to a new Hamiltonian $H^{(1)}$,
\begin{equation}
	H^{(0)} \mapsto H^{(1)}, 
\end{equation}	
in which the energy of the first fermion is modified, $H^{(1)}=H^{(0)}+\Delta a_1^{\dagger} a_1$. 
After the quench, the population of the first fermion evolves in time $t>0$ as
\begin{equation}
n_1(t)=\tr(a_1^\dagger a_1\rho(t))
\end{equation}
with $\rho(t)=\e^{-\ii H^{(1)}t}\,\rho(0)\, \e^{\ii H^{(1)}t}$.
As the Hamiltonian is quadratic, it is a problem involving free fermions and can be 
numerically simulated for very long times and system sizes (see Appendix \ref{FreeFermionsI}).

In Fig.~\ref{EquilibrationFreeFermions}, we plot the time evolution of the occupation of the first site
$n_1(t)$. 
As expected, we see that after some relaxation time $t$, $n_1(t)$ equilibrates to the value predicted by
the GGE---which is relatively far from the one given by the Gibbs equilibration model. The situation described in this example, 
a quench and the characterisation of the equilibrium state, is extensively studied in the literature, see for a recent 
review ref.\ \ \cite{Gogolin2015}. 
In order to study thermodynamic processes in which many quenches and equilibrations
are performed, it will be necessary to promote the suitability of 
effective descriptions in terms of GGE states for equilibration processes beyond a single quench.

\begin{figure}
  \includegraphics[width=1.05\columnwidth]{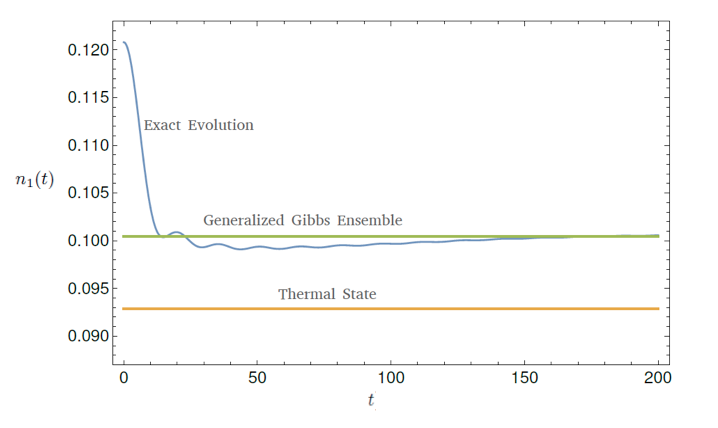}
   \caption{Time evolution of the occupation of the first site of the lattice $n_1=a_1^\dagger a_1$
   for a quadratic Hamiltonian of $n$ fermions in a one dimensional lattice. For the example we take 
   $n=100$, $\epsilon_i=1$, $\Delta=0.15$, $\beta=2$, $g=0.1$ and time is measured in units of $1/(10g)$.
  An equilibration around the GGE is observed, even for this moderately sized quantum system.}
  
\label{EquilibrationFreeFermions}
\end{figure}

\section{Framework for thermodynamic protocols}
\label{Sec:OperationsAndEquilibration}

In the previous section we have introduced the different equilibration models, given by equations \eqref{modelseq}-\eqref{modelseq3}, that describe the equilibrium state that is reached when a system initially out of equilibrium in a state $\rho$ evolves under a Hamiltonian $H$. 
One way to bring a system out of equilibrium is to quench its Hamiltonian.
More explicitly, a system initially at equilibrium with initial Hamiltonian $H^{(\text{ini})}$ undergoes a quench $H^{(\text{ini})} \mapsto H^{(\text{fin})}$ and starts to evolve non-trivially under 
the new Hamiltonian $H^{\text{(fin)}}$.
The models of equilibration introduced above can be used to describe the new equilibrium state 
that is reached after a single quench and a posterior sufficiently long time evolution under $H^{\text{(fin)}}$. 
%
However, thermodynamic processes (for instance a protocol of work extraction) often involve a series of quenches and equilibrations. We  now extend our previous considerations to such processes involving sequences of quenches and equilibrations.

\subsection{Equilibration under repeated quenches}\label{sec:equilibrationmanyquenches}
Consider a sequence of changes of the Hamiltonian, as defined by a list of $N+1$ Hamiltonians, $H^{(m)}$, where $m=0,1,\ldots,N$ denotes the step in the protocol and $H^{(0)}$ is the initial Hamiltonian. These Hamiltonian transformations $H^{(m-1)}\mapsto H^{(m)}$ are considered to be \emph{quenches}, 
in the sense that they are performed sufficiently fast such that the state of the system $\rho$ is unchanged. Let us denote the time at which the quench $H^{(m-1)}\mapsto H^{(m)}$ is performed by $t_m$ with $t_m<t_{m+1}$ for all $m$. After a quench, the system evolves under the Hamiltonian $H^{(m)}$ for a time $t_{m+1}-t_m$ until a new quench $H^{(m)}\mapsto H^{(m+1)}$ is performed at time $t_{m+1}$. This time interval is taken to be much longer than the equilibration time such that the system can be considered to be in equilibrium. 
The exact state of the system $\rho(t)$ when $m$ quenches have taken place ($t_m<t<t_{m+1}$) is given by, 
\begin{equation}\label{eq:exactevolution}
 \rho(t)=  \e^{-\ii (t-t_m) H^{(m)}} \rho(t_m) \e^{\ii (t-t_m) H^{(m)}},
\end{equation}
where $\rho(t_m)$ is the state of the system at $t=t_m$ when the 
Hamiltonian $H^{(m)}$ starts to dictate the evolution.
The state $\rho(t_m)$ is given by the recursive expression
\begin{equation}\label{eq:exactevolution2}
 \rho(t_k)=  \e^{-\ii (t_{k}-t_{k-1}) H^{(k-1)}} \rho(t_{k-1})  
 \e^{\ii (t_{k}-t_{k-1}) H^{(k-1)}}\, ,
\end{equation}
with $\rho(t_0)$ the initial state and $k=1,\ldots,m$.

Now, our aim is to construct an effective description of the whole evolution of $\rho$, in such a way that the state after the $m$-th quench and its posterior equilibration, $\rho(t)$,  can be described by an appropriate equilibrium state. 
We denote such equilibrium state that approximates the real state after $m$ quenches,  $ \rho(t)$, 
as $\omega^{(m)}_{(\cdots)}$ where $(\cdots)$ is the place holder for one of the three models of equilibration: time-average (TA), GGE or Gibbs. 
The effective description of \eqref{eq:exactevolution} is then built in a recursive way as follows, 
\begin{align}\nonumber
\omega^{(m)}_{\text{TA}} &= \ta\big(\omega^{(m-1)}_{\text{TA}},H^{(m)}\big),\\
\label{eq:equilibrationmanyquenches}
\omega^{(m)}_{\text{GGE}}&=\gge\big(\omega_{\text{GGE}}^{(m-1)},H^{(m)},\{Q^{(m)}_i\}\big) ,  \\
\nonumber\omega_{\text{Gibbs}}^{(m)}&=\gibbs\big(\omega_{\text{Gibbs}}^{(m-1)},
H^{(m)}\big).
\end{align}
Here, $\omega_{(\cdots)}^{(0)}=\rho(t_0)$ is the intial state, before any quench or evolution has taken place. 
 Note that, when constructing the GGE description, the set of conserved quantities $\{Q^{(m)}_i\}$ changes for every Hamiltonian $H^{(m)}$, as well as the Lagrange multipliers $\{\lambda_j^{(m)}\}_{j=1}^{q}$, or simply the inverse temperature $\beta^{(m)}$ in the case of equilibration to the Gibbs ensemble.

In order to provide a motivation and interpretation of eq.\ \eqref{eq:equilibrationmanyquenches}, together with the implicit assumptions that come into play, let us illustrate it with a simple example. Suppose a system initially in state $\rho(0)$ and with Hamiltonian $H^{(0)}$. At time $t_1$, we perform a first quench $H^{(0)} \mapsto H^{(1)}$ and let the system evolve under $H^{(1)}$; at time $t_2$ we perform second quench $H^{(1)} \mapsto H^{(2)}$ and let the system evolve under $H^{(2)}$
until it equilibrates at time $t$. For both evolutions, we now consider effective descriptions in terms of GGE states.
After the evolution under $H^{(1)}$ and immediately before performing the second quench, 
the system is exactly described by $\rho(t_2)$ as given by eq.\ \eqref{eq:exactevolution2}. 
For a set of conserved quantities $\{Q^{\mathrm{(1)}}_i\}$, 
the corresponding GGE equilibrium state is given by,
\begin{equation}\label{eq:GGEex}
\omega^{(1)}_{\text{GGE}}=\gge\big(\rho(t_1),H^{\text{(1)}},\{Q^{\text{(1)}}_i\}\big) \simeq \rho(t_2),
\end{equation}
where the symbol ``$\simeq$'' means in this context that the average value of relevant observables is well approximated by $\omega^{(1)}_{\text{GGE}}$, that is
\begin{equation}\label{eq:accurategeneral}
\tr(A \rho(t_2)) \simeq \tr(A \omega^{(1)}_{\text{GGE}}).
\end{equation}
Now, when describing the equilibrium state after the second quench, one can simply apply the same recipe. That is, the state $\rho(t^{(1)})$ is the initial state when the evolution under $H^{(2)}$ starts. Then, assuming that the new conserved quantities $\{Q^{(2)}_i\}_i$ are chosen appropriately and applying the same reasoning one obtains an approximation by taking
\begin{equation}
\gge\big(\rho(t_2),H^{\text{(2)}},\{Q^{\text{(2)}}_i\}\big) \simeq \rho(t),
\end{equation}
with $t$ longer than the $t_2$ plus the subsequent equilibration time.
Importantly, note that this effective description is not efficient, in the sense that it requires keeping track of the exact state $\rho(t_2)$ to obtain the equilibrium state at time $t$. If this is extended to $N$ quenches, having to keep track of the exact evolution until the $(N-1)$-th quench is as demanding as keeping track of the whole exact evolution over the process. It is here when the effective description \eqref{eq:equilibrationmanyquenches} becomes handy, as it can be constructed by keeping track of the value of the conserved quantities only. First of all, coming back to the first evolution, note that by applying  \eqref{eq:equilibrationmanyquenches} with $m=1$ we recover \eqref{eq:GGEex}, i.e., the standard result for single quenches. Now, in order to construct the GGE state corresponding to $\rho(t)$, we assume that the conserved quantities $\{Q_i^{(2)}\}$ are within the set of physically relevant observables $A$ in \eqref{eq:accurategeneral}. That is, we assume that 
\begin{equation}\label{eq:accurateconserved}
\tr(Q_i^{(2)} \rho(t_2)) \simeq \tr(Q_i^{(2)} \omega^{(1)}_{\text{GGE}})
\end{equation}
for all $i$. In this way, in order to obtain the equilibrium GGE ensemble after the second quench, it is not necessary to keep track of the exact state $\rho(t_2)$, but one can simply use $\omega_{\text{GGE}}^{(1)}$ instead. Using \eqref{eq:accurateconserved} we then obtain, 
\begin{eqnarray}
\nonumber\omega_{\text{GGE}}^{(2)} 
\nonumber&:=&\gge\big(\omega_{\text{GGE}}^{(1)},H^{\text{(2)}},\{Q^{\text{(2)}}_i\}\big)\\ &\simeq& \gge\big(\rho(t_2),H^{\text{(2)}},\{Q^{\text{(2)}}_i\}\big)\\
&\simeq& \rho(t).
\end{eqnarray}
Extending the same reasoning to the case of $N$ quenches and other models of equilibration other than the GGE, we arrive to an effective description of the form \eqref{eq:equilibrationmanyquenches}.


In the rest of this work we will always use the effective description \eqref{eq:equilibrationmanyquenches} for the full process consisting on a sequence of quenches and equilibrations. We do not claim by this that this model will accurately describe the real dynamics of any system or protocol, and indeed we explicitly leave here as an open question to identify for which Hamiltonians and conserved quantities  condition \eqref{eq:accurateconserved} is satisfied for each quench. Nonetheless, and in exactly the same way as equilibration to the Gibbs state is assumed in the usual scenario in thermodynamics, we will assume that equilibration to statistical ensembles of the form \eqref{eq:equilibrationmanyquenches} occurs over any protocol, so that we can tackle questions about entropy production and work extraction.

To examine the validity of our model, we will later provide a numerical comparison of the real exact evolution and the model of eq.\ \eqref{eq:equilibrationmanyquenches} for the case of free fermions. We will see for this example that the model predicts with great accuracy the amount of work that is extracted in a protocol involving a sequence of quenches.

\subsection{Work cost of quenches}
Concatenations of quenches and equilibrations constitute a framework to describe thermodynamic processes -see, e.g., Refs. \cite{Aberg2013a,Anders2013,StrongCoupling}. Within this framework, work is associated with the  input energy under quenches, whereas heat is associated with the exchange of energy under equilibration processes. At the level of average quantities, the work cost of a single quench, $H^{(m-1)}\mapsto H^{(m)}$, reads
\begin{equation}\label{eq:abovequantity}
W^{(m)} := \tr \left( \rho(t_m)(H^{(m)}-H^{(m-1)})\right)\, ,
\end{equation}
where $\rho(t_m)$ is given in \eqref{eq:exactevolution}.
The main assumption of this study is precisely that the work cost of a quench is very well approximated by the effective description of the equilibrium state, i.~e.
\begin{equation}
W^{(m)} = \tr  \left(\omega^{(m-1)}_{(\cdots)}(H^{(m)}-H^{(m-1)})\right)\, ,
\end{equation}
where $\omega^{(m-1)}_{(\cdots)}$ is its effective description \eqref{eq:equilibrationmanyquenches}. 
While we focus our attention on average quantities, primarily for simplicity of the analysis, one could also conceive a study of work extraction  under GGE for other work quantifiers \cite{HorodeckiOppenheim2013,Aberg2013,Gallego2015}.
As the equilibration processes happen spontaneously and have no work cost, 
the total work extracted in the entire protocol is simply given by the sum of the steps
\begin{equation}\label{eq:TotalWdef}
W := \sum_{m=1}^N W^{(m)}.
\end{equation}

\subsection{The system - bath setting beyond the weak coupling and infinite bath limits}
A particularly relevant scenario is the system - bath setting.
We call system $S$ to the part of the total system upon which
one has control and it is possible to quench its Hamiltonian $H_S$.
The bath $B$ contains the degrees of freedom upon one has no control
and it is the responsible for equilibrating the system $S$. 
In order for this equilibration to happen, 
the dimension of the Hilbert-space of $S$, $\dim(\mc{H}_{S})$, 
is considered to be much smaller than that of the bath, 
\begin{equation}
\dim(\mc{H}_{S})\ll \dim(\mc{H}_{B})
\end{equation}
and the total Hamiltonian to be of the form \cite{footnoteproduct},
\begin{equation}
H^{(m)} = H_{S}^{(m)}\otimes \id_{B} + \id_{S}\otimes H_{B} + V\, ,
\label{eq:HSB0}
\end{equation}
where the interaction $V$ is supported on $S$ and $B$ and couples the two subsystems. 
Unlike the standard assumptions in thermodynamics, 
note that we do not assume that the interaction $V$ is weak or that bath size is infinite. 
Let us be more explicit about what we mean by that.

Usually, within thermodynamics, it is assumed that the system $S$ equilibrates, upon contact with a bath $B$, according to 
\begin{equation}\label{eq:thermalequilibration}
\tr_B\left( \omega_{\beta}^{(m)}\right)=\Omega_{\beta}(H_S^{(m)}):=\frac{e^{-\beta H_S^{(m)}}}{Z},
\end{equation}
where $\beta>0$ is \emph{fixed} throughout all the protocol. 
In contrast, in the model that we  consider, given by $\omega_{\text{Gibbs}}^{(m)}$ in \eqref{eq:equilibrationmanyquenches}, the inverse temperature changes along the protocol and the Gibbs states describe the whole compound $SB$. 
Nonetheless, let us note that the model of equilibration $\Omega_{\beta}$ in \eqref{eq:thermalequilibration} represents a particular case of our Gibbsian model $\Omega_{\text{Gibbs}}$ in the limit of weak coupling and infinite bath.
In \emph{the limit of an infinite bath}, the total energy of $SB$ in \eqref{eq:equilibrationmanyquenches} will not be substantially affected by the energy pumped or subtracted in all the quenches $H_{SB}^{(m)} \mapsto H_{SB}^{(m+1)}$ and the parameter $\beta^{(m)}$ will remain constant throughout the protocol, $\beta^{(m)}\approx \beta$ for all $m$. 
In \emph{the limit of weak coupling} $V$ between $S$ and $B$, then $\Omega_{\beta^{(m)}}(H^{(m)}_{SB})\approx \Omega_{\beta^{(m)}}(H^{(m)}_{S}) \otimes \Omega_{\beta^{(m)}}(H^{(m)}_{B})$. 

In sum, the model of equilibration $\omega_{\text{Gibbs}}^{(m)}$ should be regarded as a correction to the usual setup in thermodynamics given by eq.\ \eqref{eq:thermalequilibration}. This correction incorporates the fact that the bath is of finite size, which introduces a dependence of the inverse temperature $\beta^{(m)}$ and also allows for strong couplings between $S$ and $B$.

\section{Entropy production and reversible processes}\label{sec:entropyproduction}

An important quantity in thermodynamic processes is the entropy production on system and bath during the protocol. Of course, the exact unitary 
dynamics on $SB$ does not change the von~Neumann entropy 
in the system. However, we are using an effective description on $SB$, given by \eqref{eq:equilibrationmanyquenches},
and in this effective description the entropy in the system $SB$ might well change. Indeed, 
due to the fact the equilibration models can all be understood as a maximisation of the entropy given some constraints, it follows that the entropy of the states $\omega^{(m)}$ in \eqref{eq:equilibrationmanyquenches} is non-decreasing during a protocol
\begin{equation}
S(\omega^{(m)}) \geq S(\omega^{(m-1)})\quad \forall \ \  m=1,\ldots,N. 
\end{equation}
where $S$ is the von Neumann entropy defined as
\begin{equation}
S(\rho)=-\tr(\rho \log \rho).
\end{equation}
Therefore, sequences of quenches followed by equilibrations are in general irreversible: if we start with the final state of the protocol and then run the protocol backwards, we will in general not end up with the original initial state.

From phenomenological thermodynamics we would expect that the protocols become reversible if they are done in a quasi-static way. In the context of our set of operations, a quasi-static process is defined by considering $N\rightarrow \infty$ quenches $H^{(m)} \mapsto H^{(m+1)}$ such that $H^{(m+1)}-H^{(m)}$ is of order $1/N$, followed each by an equilibration process as given by eq.\ \eqref{eq:equilibrationmanyquenches}. In this limit of an infinite number of quenches we can simply describe the quasi-static process by defining the continuous path of Hamiltonians as $u\mapsto H(u)$ with $u\in[0,1]$. This corresponds to the Hamiltonian $H^{(m)} = H(u=m/N)$, and equivalently for the equilibrium state $\omega(u=m/N) := \omega^{(m)}$ in the limit of $N \rightarrow \infty$ (where $\omega$ is an effective description of the equilibirum state given by TA, GGE or Gibbs). We will be concerned with the von Neumann entropy of the equilibrium state along the trajectory
\begin{equation}
S(u)= - \tr\big( \omega(u) \log \omega (u)\big).
\end{equation}

We now discuss in detail under which conditions the entropy remains constant over the quasi-static process, i.e. $S(1)=S(0)$,  for the three models of equilibration. 
Importantly, note that we are concerned with the entropy production in a given quasi-static process. Hence, as the quasi-static process requires an arbitrarily large number of quenches and subsequent equilibrations, it is by definition an arbitrarily slow process. We will see that the fact that the process is arbitrarily slow alone (by definition as it is a quasi-static process) does not guarantee that there is no entropy production. 

\subsection{Entropy production for time averaged ensembles}

We start by analysing the entropy production of a quasi-static process when all conserved quantities are taken into account. In this case the equilibrium state is given by $\omega_{\text{TA}}(u)$. Our first result shows that there is no entropy production in a quasi-static process if the trajectory of Hamiltonians $u\mapsto H(u)$ is smooth. 

\begin{result}[Absence of entropy production within the TA model] \label{res:entropyproductiontimeaverage} Consider a Hamiltonian trajectory $u\mapsto H(u)=\sum_k E_k(u) \proj{E_k(u)}$, and a quasi-static process along this trajectory which induces a family of time-average states $\omega_{\rm TA}(u)$.  Then, if the trajectory is continuous and the eigenvectors $\ket{E_k(u)}$ are differentiable, there is no entropy production in such a quasi-static process, that is, $S(0)=S(1)=0$.
\end{result}

The proof and discussion can be found in Appendix \ref{sec:reversible_processes}. 
Note that this result is independent of the state which is evolving under $H(u)$. In fact, for a given state, there exist quenches that are not quasi-static but preserve its entropy, such as any quench to a Hamiltonian with the same eigenbasis as the state. This is for instance the case of raising and lowering energy levels.

\subsection{Entropy production for generalised Gibbs ensembles}

Now, we consider the case of a generic GGE equilibration where not all the conserved quantities are taken into account.
In this case, the equilibration model \eqref{eq:equilibrationmanyquenches} satisfies the relation,
\begin{equation} \label{eq:paralleltransport}
\tr \left(\omega^{(m)}_{\text{GGE}} Q_i^{(m)}\right) = \tr \left(\omega^{(m-1)}_{\text{GGE}} Q_i^{(m)}\right),
\end{equation}
for all $i=1,\ldots,q$. Here the $\{Q_i^{(m)}\}$ correspond to the $q$ conserved quantities of $H^{(m)}$, 
and eq.\ \eqref{eq:paralleltransport} determines the corresponding Lagrange multipliers $\lambda_i^{(m)}$ 
in \eqref{eq:defGGE}. 
For such equilibrium states,  we also identify conditions so that there is no entropy production. More precisely, we find the following:

\begin{result}[Absence of entropy production within the GGE model]\label{res:entropyproductionGGE} Consider a quasi-static process along a Hamiltonian trajectory $u\rightarrow H(u)$ described by a family of equilibrium states $\omega_{\rm GGE}(u)$ . Then, the entropy of $\omega_{\rm GGE}(u)$ is preserved in such a quasi-static process, provided that  the Lagrange-multipliers as determined by $\eqref{eq:paralleltransport}$, form in the limit $N\rightarrow \infty$ a set of smooth functions $u\mapsto \lambda_j(u)$ for $j=1,\ldots,q$.
\end{result}

This result is shown simply by taking the continuum limit of eq.\ \eqref{eq:paralleltransport} which yields
\begin{equation}\label{eq:paralleltransportcontinuum}
\tr\left(\frac{\d \omega_{\text{GGE}}(u)}{\d u}Q_j(u)\right) = 0,\quad \forall j=1,\ldots,m
\end{equation}
which can be in turn used to show that the entropy production vanishes,
\begin{align}
\frac{\d S}{\d u}=\sum_{j=1}^m \lambda_j(u) \tr\left(\frac{\d \omega_{\text{GGE}}(u)}{\d u}Q_j(u)\right) = 0.
\end{align}
Hence, we see that, if the conditions of Result \ref{res:entropyproductionGGE} are satisfied, the entropy of the effective description in terms of GGE states is also preserved in the limit of a quasi-static process.

Let us now discuss heuristically under which conditions the premise that $\{u\mapsto \lambda_j(u)\}_1^q$ are smooth functions is expected to be fulfilled. This can be well illustrated by the following example: 

\begin{example}[Quasi-static process with entropy production within Gibbs and GGE]
Consider the case of a two dimensional system for which we take $q=1$, that is, the only conserved quantity is the 
Hamiltonian $Q_1=H$ itself (the Gibbs equilibration model). Consider initially a non-degenerate Hamiltonian $H(0)=E \ketbra{1}{1}$ and 
an arbitrary initial state $\rho(0)$ with an inverse temperature $\beta(0)>0$ and thus the entropy is smaller than $\log(2)$. 
Now suppose that the final Hamiltonian $H(1)=0$ has degenerate energy levels. 
We now show that: 
\begin{enumerate}[(i)]
\item there is a quasi-static trajectory without a smooth behaviour of the Lagrange-multipliers (in this case $\beta(u)$), 
\item  this results in a positive entropy production, and 
\item how this implies that taking only a single conserved quantity---in this case the energy---does 
not provide a good approximation of the time-averaged state. 
\end{enumerate}
\end{example}
To see the above points, take as Hamiltonian path $H(u) = E(1-u)\ketbra{1}{1} = H(0)(1-u)$ and an initial Gibbs state with inverse 
temperature $\beta(0)$. Then the eigenbasis in the entire process does not change. Now note that 
the condition \eqref{eq:paralleltransport} implies that the energy is preserved in every equilibration. 
But since we are dealing with a two-dimensional system, as long as $H(u)$ is non-degenerate, 
the state itself will remain constant $\omega(u)=\rho(u)$ for any $u \in [0,1)$. This requires that the inverse 
temperature $\beta(u) \rightarrow \infty$ as $u\rightarrow 1$: Along the path, the inverse temperature needs to fulfil 
$\beta(u) = \beta(0)/(1-u)$ to keep the state constant. Therefore, it necessarily diverges as $u\rightarrow 1$.  
To show ii), simply note that when one reaches $H(1)$, the final state is a maximally mixed state with entropy $\log(2)$, which is 
larger than the one of the initial state by assumption. To show iii), observe that the time averaged state would remain 
constant throughout the protocol, thus it differs from the GGE at $u=1$. Similar reasoning as for this 
example holds true for higher dimensional systems, where the ground state degeneracy of $H(1)$ is higher than that of $H(0)$.

The previous example shows that in some cases the premise of Result \ref{res:entropyproductionGGE} is not fulfilled, however, 
these pathological cases often imply that the chosen GGE description is not accurate. For example, in the case of encountering a
ground state degeneracy, any conserved quantity in the GGE that discerns the ground states would be enough to fix the problem. 
However, we leave in general open whether one can find smooth trajectories for $u\mapsto \lambda_j(u)$ for a given set of conserved quantities 
and trajectory of Hamiltonians---this may well depend on the specifics of the model and on the ambiguity of what constants of
motion to pick in the first place \cite{Gogolin2015}. 

\subsection{Entropy production for Gibbs ensembles}\label{sec:entropyproductiongibbs}

As discussed above in the case of the GGE ensemble, it is in general necessary to ensure that the Lagrange multipliers $u\mapsto \lambda_j(u)$ follow a smooth trajectory in order to certify that there is no entropy production in a quasi-static process. This requires to compute the Lagrange multipliers following the model of \eqref{eq:equilibrationmanyquenches} and keeping track of the conserved quantities. We will see now that the situation simplifies substantially for the case of the Gibbs model of equilibration (where the energy is the only conserved quantity). 

\begin{result}[Absence of entropy production within Gibbs model]\label{res:gibbsentropyproduction} Consider an initial and final Hamiltonian $H(0)$ and $H(1)$ and initial state $\omega_{\text{\rm Gibbs}}(0)=e^{-\beta(0) H(0)}/Z$ with finite $\beta(0)>0$. There exist a quasi-static trajectory $H(u)$ so that there is no entropy production if and only if there exist $\beta^*>0$ so that 
\begin{equation}\label{eq:entropyconstant}
S(\omega_{\text{\rm Gibbs}}(0))=S(e^{-\beta^* H(1)}/Z)
\end{equation}
\end{result}
Note that one of the implications is trivial. The final state is $e^{-\beta^{(1)} H(1)}/Z$, hence if there exist no $\beta(1)=\beta^{*}$ so that \eqref{eq:entropyconstant} is fulfilled, then it is clearly impossible to keep the entropy constant. This can happen if $H(1)$ does not admit any Gibbs state with the initial entropy. The non trivial implication of the previous result is that as long as $H(1)$ admits a Gibbs state with the initial entropy, one can always find a quasi-static trajectory that keeps the entropy constant. Indeed, we find that the quasi-static trajectory achieving it does not need to be fine-tuned. We discuss in Appendix \ref{sec:examples-constant-entropy-gibbs}, together with the proof of Result \ref{res:gibbsentropyproduction}, that any quasi-static process where the degeneracy of the ground state does not increase along the protocol will indeed keep the entropy constant. This condition is expected to be satisfied  for trajectories of generic local Hamiltonians, which have non-degenerate ground spaces 
for typical choices of the Hamiltonian parameters \cite{Aizenmann1981}. 

\subsection{Entropy production and reversibility}
We  now connect entropy production to reversibility of processes.
First, let us note that for the GGE equilibration model (similarly for the Gibbs model since it is a particular case of the former), condition \eqref{eq:paralleltransportcontinuum} is invariant if one reverses the process. More specifically, given $H(u)$ and $\omega_{\text{GGE}}(0)$ as initial state, condition \eqref{eq:paralleltransportcontinuum} determines the trajectory of states (if the premises of Result \ref{res:entropyproductionGGE} are met) $\omega_{\text{GGE}}(u)$, with $u$ from $0$ to $1$. Now, we can consider the trajectory $H(\tilde{u})$ with initial state $\omega_{\mathrm{GGE}}(\tilde{u}=0)$ with $\tilde{u}=1-u$. One can easily verify that 
\begin{equation}
\tr\left(\frac{\d \omega_{\text{GGE}}(\tilde{u})}{\d \tilde{u}}Q_j(\tilde{u})\right) = 0,\quad \forall j=1,\ldots,m.
\end{equation} 
Hence, the equilibrium state for the trajectory $H(\tilde{u})$ is given exactly by $\omega_{\text{GGE}}(\tilde{u}=1-u)$ and thus, the protocol is reversible. In other words, we have seen that for the GGE equilibration model reversible protocols correspond to arbitrarily slow protocols where no entropy is produced on the \emph{system and bath} together, exactly as is the case for 
phenomenological thermodynamics.

A well-known feature of phenomenological thermodynamics is that reversible transformations are always beneficial in work-extraction protocols, a phenomenon which is referred to as
the \emph{minimum work principle}. We will later see that this principle naturally holds when the model of equilibration is given by Gibbs states, but its range of applicability  is considerably reduced when the equilibrium states are described by GGE. Indeed, we will see explicitly that when the equilibration model is given by a GGE ensemble of free fermions, it can well be beneficial to go through a given protocol quickly and thereby producing entropy.


Before we go on to discuss explicit work extraction protocols, let us stress that the entropy in $SB$, which can only increase or remain constant, is not 
simply the sum of the entropies of $S$ and $B$. This happens because we are considering interacting quantum systems that show correlations between $S$ and $B$. 
This is true both in the exact and the effective description. Indeed, in general the von~Neumann entropy in $SB$ is smaller than or equal to the sum of local entropies
\begin{equation}
S\left(\omega\right) \leq S\left(\tr_B\left(\omega\right)\right) + S\left(\tr_S\left(\omega\right)\right),
\end{equation}
with equality if and only if $\omega=\tr_B\left(\omega\right)\otimes \tr_S\left(\omega\right)$, i.e., when $S$ and $B$ are completely uncorrelated. 
Thus, entropy-production in our set-up does not always mean that entropy is \emph{locally} produced in the system and the bath. The generation of entropy is not always 
associated with the generation of correlations, as in ref.\ \ \cite{Esposito2010}, but rather to the mixing induced by equilibration processes. The global entropy may, for example, 
increase due to a decrease of correlations, but entirely without changing the local states of the system.

As a final remark, note that in the so-called Isothermal Reversible Process (IRP) the entropy of the system $S$ does not remain constant, while the entropy of the whole compound $SB$ does, as we discuss in Examples 1 and 2 in Appendix \ref{sec:examples-constant-entropy-gibbs}).


\section{The minimum work principle and work extraction}\label{sec:workextraction}
In order to study work extraction, we first focus on the \emph{minimum work principle}, which is intimately related  to work extraction and other tasks in thermodynamics such as, e.g., the erasure of information (Landauer's Principle). 
We  take as the definition of the minimum work principle that, \emph{given an initial equilibrium state and a path of Hamiltonians,
the work performed on the system is minimal for the slowest
realisation of the process} \cite{Allahverdyan2005}. More precisely, we consider a trajectory of Hamiltonians $u\mapsto H(u)$ with $u\in [0,1]$. Consider now protocols with $N$ quenches (each followed by an equilibration). That is, we choose $N$ values $(u^1,\ldots,u^N)$, so that the protocol is determined by $H^{(m)}=H(u^m)$ and $\omega_{(\cdots)}^{(m)}$ as determined by \eqref{eq:equilibrationmanyquenches}. The minimal work principle states that the optimal protocol maximizing $W$ in \eqref{eq:TotalWdef} is the one where $N \rightarrow \infty$ and 
$u^m=m/N$. Note that here, as we generically take the convention that work is \emph{extracted} from the system, minimising the work cost corresponds to maximising $W$ in \eqref{eq:TotalWdef}.

We note that, while being the most relevant notion of the minimum work principle for our set-up (see also ref.\  \cite{Allahverdyan2005}), this definition differs from the one usually found in thermodynamics text-books, where the content of the minimal work principle reads: among all the possible paths between two \emph{fixed} equilibrium \emph{states}, reversible protocols are optimal. Here, we fix instead a given trajectory between an initial and final Hamiltonian and question whether the quasi-static realisation is also the optimal. Note that both notions---where the initial and final states are fixed or where the trajectory is fixed instead---coincide in the model of equilibration of eq.\ \eqref{eq:thermalequilibration}, which is the standard one in text-book thermodynamics. The reason is that in the model \eqref{eq:thermalequilibration} Hamiltonians and states are in one to one correspondence and all the quasi-static trajectories between two Hamiltonians provide the same work. However, when other models of equilibration are considered ($\omega_{(\cdots)}$ in \eqref{eq:equilibrationmanyquenches}) the equivalence breaks down since Hamiltonians and states are not in one-to-one correspondence: the final state depends on the specific trajectory.

It seems then natural to ask what justifies our definition of the minimal work principle. The answer lies in the fact that the notion of the minimal work principle considered here can be easily connected with the second law of thermodynamics, formulated as: no positive work can be extracted in a cyclic process from states initially in thermal equilibrium (a Gibbs state), or more generally, in a passive state. In this context cyclic refers to the fact that the initial and final Hamiltonian coincide, which does not imply that the initial and final state coincide, unless we would be using the model of equilibration \eqref{eq:thermalequilibration}. This relation with the minimal work principle and the second law will be made explicit in the following sections where we study the Gibbs and the time-average ensembles.

\subsection{The minimum work principle and work extraction for Gibbs ensembles}\label{sec:gibbsMinimalWorkPrinciple}

Let us consider the same setup as the one laid out in Section \ref{sec:equilibrationmanyquenches}, with an initial state $\omega_{\text{Gibbs}}^{(0)}$ and a protocol that performs $N$ quenches according to a certain trajectory $u\mapsto H(u)$, where $H^{(m)}:=H(m/N)$. Let us stress that here we do not take the limit of $N \rightarrow \infty$ and we keep it general by considering finite $N$. Let us recall from eq.\ \eqref{eq:TotalWdef} that the total work performed is given by the sum of the individual work $W^{(m)}$ in the $m$-th step,
\begin{eqnarray}
\nonumber W&=&\sum_{m=1}^{N}W^{(m)} = \sum_{m=1}^{N} \Tr \left( \omega_{\text{Gibbs}}^{(m-1)} (H^{(m-1)}- H^{(m)}) \right)\\
\nonumber &=& \tr\left(\omega_{\text{Gibbs}}^{(0)}H^{(0)}\right)- \tr\left(\omega_{\text{Gibbs}}^{(N)}H^{(N)}\right)\\
\label{eq:workdifferenceofenergies}
&+&\sum_{m=2}^{N} \Tr \left( (\omega_{\text{Gibbs}}^{(m-1)}- \omega_{\text{Gibbs}}^{(m)})H^{(m)} \right),
\end{eqnarray}
where in eq.~(\ref{eq:workdifferenceofenergies}) we have simply reorganised the terms and added and subtracted the quantity $\tr(\omega_{\text{Gibbs}}^{(N)}H^{(N)})$. 
We can now use our model of equilibration, as given by eq.\ \eqref{eq:equilibrationmanyquenches} that we recall here for completeness, 
\begin{equation}\label{eq:i-thgibbs}
\omega_{\text{Gibbs}}^{(m)}=\gibbs(\omega_{\text{Gibbs}}^{(m-1)},H^{(m)})
=\frac{e^{-\beta^{(m)} H^{(m)}}}{Z^{(m)}}
\end{equation}
for all $m\ge 1$, where $Z^{(m)}=\tr (e^{-\beta^{(m)} H^{(m)}})$ and $\beta^{(m)}>0$ is determined by the conservation of average energy: 
$\tr (\omega_{\text{Gibbs}}^{(m-1)}H^{(m)})=\tr (\omega_{\text{Gibbs}}^{(m)}H^{(m)})$. One can easily check that energy conservation implies that the last sum in \eqref{eq:workdifferenceofenergies} vanishes, which implies that 
\begin{equation}\label{eq:workinitialfinal}
W=\tr\left(\omega_{\text{Gibbs}}^{(0)}H^{(0)}\right)- \tr\left(\omega_{\text{Gibbs}}^{(N)} H^{(N)}\right),
\end{equation}
where $\omega_{\text{Gibbs}}^{(N)}$ depends on $N$ and the trajectory $H(u)$.

From eq.~\eqref{eq:workinitialfinal} we see that given a fixed final Hamiltonian $H(1)$, the protocol that costs the minimum amount of work (and maximises the extracted work $W$) is given by the one that leaves the final state with the least average energy. Since the average energy is monotonic with the entropy for Gibbs states of positive temperature, we conclude that the optimal protocol is the one minimising the entropy of the final state $\omega_{\text{Gibbs}}^{(N)}$. Furthermore, as the entropy can only increase throughout the protocol (see Sec.\  \ref{sec:entropyproduction}), a protocol creating no entropy is optimal. 

It has to be stressed that this holds true only as long as the final temperature of the Gibbs state is positive, which happens if 
\begin{equation}\label{eq:condpositivetemp}
\tr (\omega_{\text{Gibbs}}^{(N)}H^{N} ) \leq \frac{1}{d} \tr  (H^{(N)}),
\end{equation} 
where $d$ is the dimension of the Hilbert space. Note that the right hand side of the equation typically (e.g., for many body systems with short range interactions) grows linearly with the number of particles. Therefore, if the total system is big enough, we expect condition \eqref{eq:condpositivetemp} to be satisfied, and thus the minimum work principle to hold.   

Taking together the facts that a protocol creating no entropy is optimal and the results of Appendix \ref{sec:examples-constant-entropy-gibbs}--which show conditions so that the quasi-static entropy has no entropy production---one can conclude that the minimal work principle is satisfied for any trajectory so that $d_g(H(0)) \geq d_g(H(u)) \geq d_g(H(1))$ (see Result \ref{example1}). As mentioned in Sec. \ref{sec:entropyproductiongibbs}, this condition is satisfied for trajectories of generic Hamiltonians, which have non-degenerate ground spaces.  

Let us now comment on the relation between the minimal work principle and the second law of thermodynamics. First, note that if we fix a trajectory $H(u)$ so that $H(0)=H(1)$, then the final state is a Gibbs state $\Omega_{\beta_N}(H(0))$. The inverse temperature $\beta^N$ at the end of the protocol certainly depends on the particular trajectory and the number of quenches performed. However, it is clear by the discussion of Sec. \ref{sec:entropyproduction} that $S(\Omega_{\beta_N}(H(0)))\geq S(\omega(0))$. Hence, since the final state is a Gibbs state with respect to $H(0)$ and with more entropy than the initial Gibbs state and the energy is monotonic with the entropy for Gibbs state, the extracted work is negative. Note that this depends crucially on having Gibbs states as equilibrium states and it will not be reproduced by time-average or GGE models of equilibration as we discuss in the next sections. 

Also, the minimum work principle can be used to study work-extraction protocols from \emph{non-equilibrium} states. As an example,  consider as initial conditions  a pair of state and Hamiltonian $\rho^{(0)}$ and $H^{(0)}$ respectively. The goal is to extract work from $\rho^{(0)}$  by performing a \emph{cyclic} protocol, where $H^{(N)}=H^{(0)}$.  Note that here the initial state is not in a Gibbs state with respect to the initial Hamiltonian, $H^{(0)}$. Nevertheless, after the first quench, it does thermalise to $\omega_{\text{Gibbs}}^{(1)}=\Omega_{\text{Gibbs}}(\rho^{(0)},H^{(1)})$. From that moment onwards, the minimum work principle can be used, implying that it is always optimal to come back to $H^{(0)}$ by a protocol that does not create entropy.  The only remaining question is in fact to which Hamiltonian the first quench is performed, an issue that is discussed in Appendix \ref{sec:gibbsworkextraction}.

\label{sec:timeaveragework}
\subsection{Work extraction and the minimum work principle for time averaged states}\label{MinimalWorkPrincipleTA}
We  now discuss the minimum work principle for protocols of work extraction when the model of equilibration that is used is the time-average $\Omega_{\text{TA}}$. Let us assume a smooth trajectory of Hamiltonians $H(u)$ and some initial equilibrium state $\omega_{\text{TA}}(0)$. Since the trajectory of Hamiltonians is smooth we know that the final state in the quasi-static protocol $\omega^{\text{q.s.}}_{\text{TA}}(1)$ has the same entropy as the initial state $\omega_{\text{TA}}(0)$, indeed even the same eigenvalues as $\omega_{\text{TA}}(0)$ (see Appendix~\ref{sec:reversible_processes}). The question is, whether this also implies that the quasi-static protocol is optimal in terms of the average work-cost. We will show that this is in general only true if this final state in the quasi-static protocol is also a \emph{passive} state, meaning that it is diagonal in the energy-eigenbasis and the energy-populations decrease with increasing energy:
\begin{equation}
\tr(H(1)\omega^{\text{q.s.}}_{\text{TA}}(1))= \sum_k  (\omega^{\text{q.s}}_{\text{TA}}(1))^{\downarrow}_k E_k(1),
\end{equation}
where $(\omega_{\text{TA}}(1))^{\downarrow}$ is the vector of eigenvalues of $\omega_{\text{TA}}(1)$, ordered such that $(\omega_{\text{TA}}(1))^{\downarrow}_k \geq (\omega_{\text{TA}}(1))^{\downarrow}_l$ if $E_k(1)\leq E_{l}(1)$. 

\begin{result}[Passiveness of optimal protocols]\label{result:time-average-quasi-static}
Given an initial state and a smooth trajectory of Hamiltonians,
if the final state in the quasi-static realisation of the protocol is
passive, then the the quasi-static realisation of the protocol is optimal.
\end{result}
This result follows, because passive states can only \emph{increase} their average energy under any unitary transformation \cite{Pusz1978,Lenard1978}:
\begin{align}
\tr(H\rho) \leq \tr(HU\rho U^\dagger),\quad \rho\text{ passive w.r.t. }H.
\end{align}
In particular the final state of the quasi-static protocol $\omega_{\text{TA}}^{\text{q.s.}}(1)$ is related to the initial state by some unitary transformation $U^*$ since their spectra are identical. To see that the quasi-static realisation of the protocol is optimal in this case, let us now consider any realisation of the protocol with only a finite number of quenches $N$ and let us denote the final state in a protocol with $N$ quenches as $\omega_{\text{TA}}^{N}$. Since the time-average equilibration model can be thought of as applying a mixture of unitaries (evolving the system for some random time) in any finite realisation including $N$ quenches, the final state $\omega_{\text{TA}}^N$ is related to the initial state by:
\begin{align}
\omega_{\text{TA}}^N = \sum_i p_i U_i\omega_{\text{TA}}(0) U_i^\dagger = \sum_i p_i (U_iU^*)\omega_{\text{TA}}(1)(U_iU^*)^\dagger,\nonumber
\end{align}
where $p_i$ is some probability distributions of unitaries. But since $\omega_{\text{TA}}^{\text{q.s.}}(1)$ is passive, we henceforth have
\begin{align}
\tr(H(1)\omega_{\text{TA}}^N) \geq \tr(H(1)\omega_{\text{TA}}(1)),
\end{align}
which proves the claim.

The minimum work principle for cyclic unitary processes was studied in ref.\ \cite{Allahverdyan2005} where it was shown that the minimum work-principle holds if: i) the initial state is passive with respect to the initial Hamiltonian $H(0)$ \emph{and} ii) the trajectory of Hamiltonians is such that the initial and final Hamiltonians $H(0)$ and $H(1)$, respectively, do not have a level-crossing w.r.t to each other. Here, by an absence of level-crossing we mean that if $E_i(0) \geq E_j(0)$, then also $E_i(1) \geq E_j(1)$ (note that the labelling of the energy-basis is fixed since we require the Hamiltonian trajectory to be smooth).
It is now easy to see that under the premise that the initial state is passive, the condition that the final state in the quasi-static realisation is also passive is indeed equivalent to the absence of such a level-crossings. Thus, our result naturally generalises that of \cite{Allahverdyan2005}.

Finally, let us note that given two Hamiltonians $H(0)$ and $H(1)$ and an initial equilibrium state, it is always possible to construct a smooth trajectory of Hamiltonians that connects the two Hamiltonians and such that the final state in the quasi-static protocol is passive and has the same spectrum as it had initially. This can be done with the protocol presented in Appendix \ref{app:protocol_qs_topassive}. However, note that this protocol requires global control over the Hamiltonians. Once we can only control some part of the Hamiltonian, all the available smooth trajectories might lead to a non-passive final state in the quasi-static realisation, so that it can become beneficial to use a protocol with a finite number of quenches which results in entropy-production.  

As in the case of the Gibbs equilibration model, one can easily relate the analysis above to discuss the second law of thermodynamics. First, note that the optimal protocol between $H(0)$ and $H(1)$ is such the final state has the same spectrum and it is passive. Hence, if $H(0)=H(1)$ we conclude that one can extract positive work from the initial equilibrium state $\omega_{\text{TA}}(0)$ if and only if it is not passive. Of course this fact is well-known if we consider protocols of work extraction that just apply a unitary transformation to the initial state. Here, we are deriving a similar behaviour with families of protocols that are instead quenches and equilibrations to the time-average state. 

In summary, we have identified conditions that ensure that the quasi-static realisation of a given protocol is optimal. This condition generalise the ones found in ref.\  \cite{Allahverdyan2005}. Also, we have shown that any state can be brought to its passive form---keeping the same spectrum---by applying a quasi-static protocol over a specific trajectory of Hamiltonians. Altogether, this show that quasi-static protocols are as powerful for work extraction as they can conceivably be.

\subsection{Work extraction and the minimum work principle for GGE states}\label{Sec:WEGGE}

In this section we briefly analyse notions of work extraction in the case of GGE models of equilibration. Although it is difficult to provide general results for the case of the GGE, without having specified the particular form of the conserved quantities, we do include here a general formulation of the problem at hand as an introduction to particular example of free fermions that we study later. 
In this situation, the equilibrated states are maximum entropy states 
\begin{equation}
\gge(\rho,H^{(m)},\{Q_j^{(m)}\}) := \argmax_{\sigma\in\mathcal{E}(\rho,\{H^{(m)},Q_j^{(m)}\})} S(\sigma)\,.
\end{equation}
for a collection of constants of motion $\{Q_j^{(m)}\}$ that are relevant at a given step $m$ of the protocol. 
For a given protocol, the work extracted is again
\begin{equation}
W = \sum_{m=1}^{N} W^{(m)} = \sum_{m=1}^{N} \tr\left(
\omega^{(m-1)}_{\text{GGE}}(H^{(m-1)} - H^{(m)})
\right),
\end{equation}
so that in order to compute the extracted work for a given protocol, one has to keep 
track of the Lagrange multipliers along that protocol. 
The optimal work extraction is attained as the supremum of this expression over such protocols.
In agreement with our considerations for time-averaged states, here we will find that the minimum work principle is in general not satisfied for many-body models that equilibrate to a GGE. Ultimately, this result is linked to the fact that \emph{for GGEs there is in general no direct link between entropy and energy}, in strong contrast to the case of Gibbs states. We show this statement by considering specific classes of models for which the GGE is relevant, namely the class of physical systems described by free fermions, a most relevant type of systems that are known to be well described by the generalised Gibbs ensemble. In particular we will show an example where a fast protocol outperforms a slow protocol despite the fact that an effective description by Gibbs states would suggest the opposite.

\section{Free fermionic systems}
\label{sec:freefermions}

On top of showing the validity of the above result, the reason for largely focusing on quadratic fermionic models is three-fold. First, they can be efficiently simulated, allowing us to test how well
the effective description of the system approximates its real (exact) dynamics. Also, they are integrable, which implies that a GGE description is in general necessary to 
capture their equilibration behaviour \cite{PolkovnikovReview,Gogolin2015}. Finally, they can be simulated with
ultra-cold atoms in optical lattices in and out of equilibrium \cite{Koehl2,Schneider_fermionic_transport,BlochReview,Lewenstein2007}. 
While the discussion presented here is focused on non-interacting fermionic systems, it should be clear that their bosonic lattice instances
\cite{BlochSimulations,BlochReview,nature_bloch_eisert}
and even bosonic continuous systems \cite{SchmiedmayerGGE,SchmiedmayercMPS} can be captured in an analogous framework
with very similar predictions. The latter situation is specifically
interesting as modelling the physics of ultra-cold atoms on atom chips that is expected to provide an experimental
platform probing the situation explored here where a GGE description is relevant.

%
\subsection{Hamiltonian, covariance matrix and GGE construction}
We consider quadratic fermionic Hamiltonians of the form 
\begin{equation}
	H=\sum_{i,j=1}^nc_{i,j} a_i^{\dagger} a_j, 
\end{equation}
where $n$ is the number of different modes and the fermionic operators satisfy the anti-commutation relations $\{a_i,a_i^{\dagger}\}=\delta_{i,j}$, $
\{a_i,a_j\}=\{a_i^{\dagger},a_j^{\dagger}\}=0$. The Hamiltonian $H$ can be transformed into
\begin{equation}
H=\sum_{k=1}^n \epsilon_k \eta_k^{\dagger} \eta_k\, ,
\label{Hfreeferm}
\end{equation}
where $\eta_k^{(\dagger)}$ is the annihilation (creation) operator corresponding to the
$k$-th eigenmode of the Hamiltonian.

It is well known that equilibrium states of Hamiltonians of the form \eqref{Hfreeferm} are not well described by Gibbs states, but 
rather by generalised Gibbs ensembles, 
with the conserved quantities being the occupations of the energy modes $Q_k=\eta_k^{\dagger} \eta_k$, $k=1,\dots,n$ \cite{Gogolin2015}. 
Notice that the number of conserved quantities used for the construction of the GGE is  the number of distinct modes $n$ and, hence, 
is linear (and not exponential) in the system size.

We define the \emph{correlation matrix} $\gamma(\rho)$ of a state $\rho$ as the symmetric matrix having entries 
\begin{equation}
\gamma_{i,j}(\rho) = {\rm Tr} ( \eta_i^{\dagger} \eta_j \rho ).
\label{gammai,j}
\end{equation}
If the state $\rho$ is Gaussian, then  $\gamma(\rho)$ contains all information about $\rho$, and its time evolution under Hamiltonians of the type \eqref{Hfreeferm} keeps it Gaussian. In other words, the full density matrix $\rho$ can be reconstructed from just knowing the correlation matrix. 

The correlation matrix of the GGE $\gge(\rho,H,\{\eta^{\dagger}_k \eta_k\})$ is found by maximising the entropy while preserving  all $Q_k=\eta_k^{\dagger} \eta_k$, which simply reduces to dephasing the correlation matrix defined in \eqref{gammai,j} to the diagonal (see Appendix \ref{FreeFermionsI} for details). This provides a simple method for obtaining $\gamma(\gge(\rho,H,\{\eta^{\dagger}_k \eta_k\}))$. 

Note that these GGE descriptions are also Gaussian states. 
Hence, in the following we can always restrict to Gaussian states. Even when the initial state is not Gaussian, all the results are unchanged if the initial state is replaced by a Gaussian state that has the same correlation matrix.
Consequently, in the following, the discussion is reduced to the level of correlation matrices instead of the full density matrices. This allows us to perform numerical simulations of the real time-evolution as well as the effective description of large systems, since they have dimension $n\times n$ instead of the $2^n\times 2^n$ needed to describe the full density matrix.

\subsection{Work extraction and minimum work principle for free fermions}\label{Sec:Optimal protocols for work extraction and minimum work principle}

First we consider optimal protocols for work extraction in a scenario where the Hamiltonian can be transformed to any 
quadratic Hamiltonian of the form \eqref{Hfreeferm}. The discussion is similar to that of  Sec.\ \ref{MinimalWorkPrincipleTA}, but in the context of GGE equilibrium states. 
As in the previous sections, the optimal protocol is the one minimising the final energy, $\tr(\omega^{(N)}_{\text{GGE}} H^{(0)})=\sum_k n_k^{(N)} \epsilon_k^{(0)}$, where we assume the process to be cyclic and $n_k^{(N)}=\Tr({\eta_k^{(0)}}^{\dagger} \eta_k^{(0)} \omega^{(N)}_{\text{GGE}})$.
In Appendix \ref{FreeFermionsIII}, we show that this minimisation yields $\tr(\omega^{(N)}_{\text{GGE}}H^{(0)}) \geq \tr(\omega^{*}_{\text{GGE}}H^{(0)})$, with
\begin{equation}
 \tr(\omega^{*}_{\text{GGE}}H^{(0)})=\sum_{k=1}^{n} (d^{(0)})^\downarrow_k (\epsilon^{(0)})^\uparrow_k,
 \label{pord}
\end{equation}
where $d^{(0)}_k$ are the eigenvalues of $\gamma( \rho^{(0)})$ and the symbols ${}^\uparrow$ and ${}^\downarrow$ indicate that the lists are ordered in increasing and decreasing order, respectively. An explicit protocol saturating this bound is constructed in Appendix \ref{FreeFermionsIII}. The optimal protocol is found to be $\emph{reversible}$, so that no entropy is generated, and one needs to perform an arbitrarily large amount of quenches to reach optimality. 

In the optimal final state $\omega^{*}_{\text{GGE}}$, the diagonal elements of the correlation matrix, corresponding to the population of the energy modes, decay as the energy of the modes increases. This form is reminiscent of the passive states previously introduced.  However, in general, states of the form $\omega^{*}_{\text{GGE}}$ do not need to be passive:  
While in passive states the occupation probabilities of the energy eigenstates are decreasing with increasing energy, here only the occupations of the different fermionic \emph{modes} decrease with the energy of the mode. The total energies are however obtained by combinations of different modes. An example of a state that is non-passive, but where the mode-populations are decreasing with increasing mode-energy is provided in Appendix  \ref{FreeFermionsIII}.


Regarding the minimum work principle, one can use a similar line of reasoning as in Sec.\ \ref{MinimalWorkPrincipleTA}. For a fixed process, the minimum work principle is guaranteed to hold true as long as the possible final states---which are realised by implementing the process at different speeds---have the form  \eqref{pord}, i.e., their populations decrease with the energy of the modes. If this condition is not satisfied, the minimum work principle does not  hold in general.

\subsection{Numerical results: comparison between exact dynamics and effective descriptions}\label{sec:numerical}
In this section we compute the work extracted in different scenarios by (i) a numerical simulation of the exact unitary evolution of the system, (ii) using the effective description in terms of Gibbs states, and (iii) in terms of GGE states. 
As physical system we consider a chain of fermions, taking as an initial Hamiltonian,
\begin{equation}\label{eq:Hfermchain}
H^{(0)}= \sum_{i=1}^{n}  \epsilon_i a_i^{\dagger} a_i +   g\sum_{i=1}^{n-1}  \left(a_i^{\dagger} a_{i+1}+a_{i+1}^{\dagger} a_{i}\right).
\end{equation}
First we study the optimal protocol for the case unrestricted Hamiltonian case derived in Appendix\ \ref{Sec:WEGGE}, and next we consider the case of local changes of the Hamiltonian. In all cases we find a very good agreement between the real dynamics and the GGE effective description.  

Besides comparing the effective descriptions with the real dynamics, we also study the applicability of the minimum work principle. We give an explicit example of a process in which producing entropy is beneficial for work extraction, hence showing that the minimum work principle is violated in this example.

\subsubsection{Work extraction with unrestricted Hamiltonians and free fermions}

Here, we take as the initial state $\rho^{(0)}$ a GGE state whose populations $\gamma(\rho^{(0)})_{i,i} \in (0,1)$ in \eqref{gammai,j} are chosen i.i.d.\ from a Gaussian distribution. Again, note that this state is Gaussian. We then apply the protocol described in Appendix \ref{FreeFermionsIII} for maximal work extraction, and compare the results obtained by the exact dynamics and the GGE model of equilibration.  The exact dynamics are computed by, after the the $i$-th quench, letting the system unitarily evolve under the Hamiltonian $H^{(i)}$ for a time much longer than the time scale of equilibration. 
Fig.\ \ref{figGlobalQuenches} shows the results obtained using both approaches. It shows a very good agreement, as long as the number of fermions is sufficiently large (in the figure $n=100$). Yet small discrepancies are observed, which is due to the fact that we implement global quenches, for which the state  may not equilibrate. 
Note that, when performing local quenches and starting with a Gibbs state, as in Fig.\ \ref{figLocalQuenches}, equilibration of local observables is expected (see Sec.~\ref{sec:equilibrationmodels}) and the agreement is indeed excellent.

We can also see in Fig.\ \ref{figGlobalQuenches} how work increases as the process becomes slower, becoming maximal in the limit $N\rightarrow \infty$, when reversibility is achieved. This is in agreement with our considerations in Sec. \ref{Sec:Optimal protocols for work extraction and minimum work principle}.

\begin{figure}
   \includegraphics[width=1.\columnwidth]{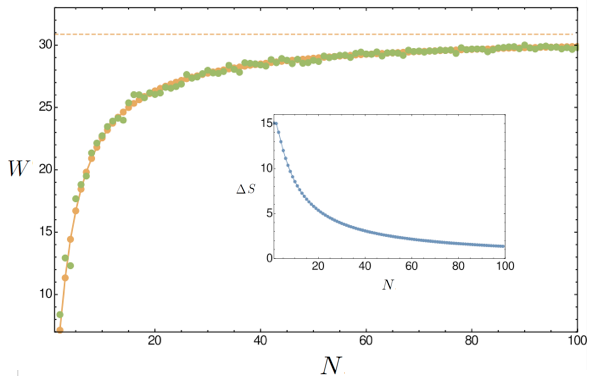}
    \caption{Extracted work in the optimal protocol with unrestricted Hamiltonians. As an initial state, we take a diagonal state in the basis $H^{(0)}$, with the populations $\{p_k^{(0)}\}$ 
    chosen at random between $0$ and $1$. We take $\epsilon=1, g=0.8$ and $N=100$. In order to simulate the real dynamics, after every quench, we let the system evolve for a time chosen at random between $20/g$ and $100/g$. In green, we show the results using the actual unitary dynamics, in yellow our effective description in terms of GGE states, and in dashed lines the analytical result leading to eq.\ \eqref{pord}. The inset figure shows the entropy generated in the effective description using GGE states. As the number of quenches increases (i.e., the process becomes slower), the generated entropy tends to zero and the extracted work tends to the upper bound. 
    }
\label{figGlobalQuenches}
\end{figure}

\subsubsection{Work extraction with restricted free fermionic Hamiltonians with a Gibbs initial state}
\label{Sec: Work extraction with restricted Hamiltonians}

Let us now assume that the Hamiltonian can only be locally modified, as discussed in Sec.\  \ref{Sec:OperationsAndEquilibration}. The Hamiltonian \eqref{eq:Hfermchain} is split  in three components. $H_S=\epsilon_1 a^{\dagger}_1 a_1$ ($S$ is a single fermion), $V=g (a_1^{\dagger} a_{2}+a_{2}^{\dagger} a_{1})$ and $H_B=H-V-H_S$. Our capability to change the Hamiltonian is thus reduced to a single parameter: the local energy $\epsilon_1$. Note that the coupling between the $S$ and the 
$B$ is not assumed to be weak.
The initial state takes the form,
\begin{equation}
\rho^{(0)} = \rho_S\otimes \Omega_\beta(H_B) = \rho_S \otimes \frac{e^{-\beta H_B}}{Z} ,
\label{initialStateG}
\end{equation}
where $\rho_S$ is initially out of thermal equilibrium; for example, in Fig.\ \ref{figLocalQuenches}, it is set to a lower temperature than the bath. As discussed before, we do not need to assume that the initial state $\rho_S$ (and hence $\rho^{(0)}$) is Gaussian, but the work extracted will only depend on its correlation matrix and not on the full density matrix, since the energy is a sum of second moments of the fermionic creation and annihilation operators and all the GGE states constructed in the process are Gaussian automatically.

Fig.\ \ref{figLocalQuenches} shows the extracted work from $\rho_S$ as a function of the number of quenches $N$, which is computed using the real exact unitary evolution, and the effective description in terms of both GGE and Gibbs states. The agreement between the unitary dynamics and the GGE description is excellent, for any value of $N$ and the parameters, but the Gibbs states fail to describe the process. Even if the bath is initially in a Gibbs state, see eq.\ \eqref{initialStateG}, the posterior evolution of $SB$ 
can not be correctly described by them.
Although the description in terms of Gibbs ensembles is quantitatively incorrect, it is fair to say that it describes some qualitative features of the results. In particular, the exact dynamics satisfies the minimum work principle, and so does the effective description with Gibbs states. This follows because condition \eqref{eq:condpositivetemp} is satisfied during the process. However, as we show in the next section, condition \eqref{eq:condpositivetemp} can fail to predict the applicability of the minimum work principle.

\begin{figure}
   \includegraphics[width=1.\columnwidth]{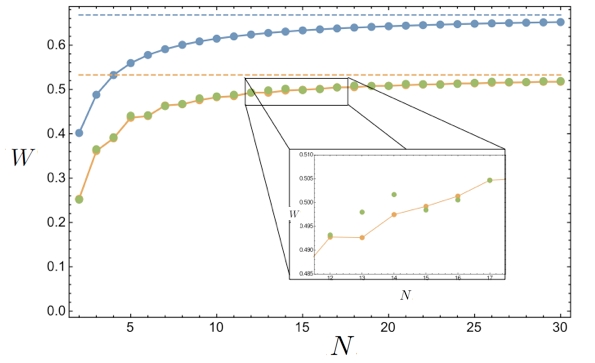}
    \caption{Extracted work with only local transformations on the state of the system. The different points correspond to the exact unitary evolution (in green), to the effective evolution in terms  GGE states (in yellow), and the effective evolution using Gibbs states (in blue). The continuous lines correspond to transformations with $N \rightarrow \infty$. As an initial state we take, $\beta=1/2$, $\tr (a_1^{\dagger} a_1 \rho_S )=0.1$, $n=100$. For the initial Hamiltonian, $\epsilon_0=0.1$, $\epsilon_i=1$ $\forall i\neq 1$, $g=0.5$. As a protocol we perform a first quench to $\epsilon_1=4.3$, followed by $N-1$ equidistant quenches back to the original Hamiltonian.  As in Fig.\ \ref{figGlobalQuenches}, the exact evolution is obtained by letting system and bath interact for a time much larger than the equilibration time ($t_{\text{Eq}} \propto 1/g$).
    }
\label{figLocalQuenches}
\end{figure}

\subsubsection{Work extraction with free fermionic restricted Hamiltonians with a GGE initial state}
\label{Sec: Work extraction with restricted Hamiltonians GGE numerics}

Equilibrium states when dealing with Hamiltonians of the type  \eqref{Hfreeferm} are well described by GGE states, it is therefore 
natural to generalise the initial state \eqref{initialStateG} to
\begin{equation}
\rho_0=\rho_S \otimes \omega^{(B)}_{\text{GGE}},
\label{initialStateGGE}
\end{equation}
where $\omega^{(B)}_{\text{GGE}}$ is a GGE state with respect to the local Hamiltonian of $B$, 
$H_B=\sum_{k=1}^n \epsilon^{(B)}_k \eta_k^{(B)\dagger} \eta^{(B)}_k.$
Let us now pick a very particular initial state given by
\begin{equation}
 {\rm Tr} (\omega^{(B)}_{\text{GGE}}  \eta_k^{(B)\dagger} \eta_k^{(B)}) = \left\{
     \begin{array}{lr}
       1 & k\geq K\\
       0 &  k< K 
     \end{array}
   \right. ,
\end{equation} 
for some $K<n$. That is, only the $K$ most energetic modes are populated. No actual thermal state with positive temperature would have such properties due to the population inversion of the fermionic modes. It is important to acknowledge, however, that if we would chose an effective description as a Gibbs state for such initial states, we would nevertheless obtain a positive effective temperature provided that condition \eqref{eq:condpositivetemp} is satisfied. This will be the case as long as the number of populated energy-levels $K$ is small enough. Indeed, for any finite $K$, but large $n$, the energy-density in the state is much lower than the critical energy-density needed for negative effective temperatures.

The work extracted in a particular protocol with initial state \eqref{initialStateGGE} is plotted in Fig.\ \ref{figLocalQuenchesII}.  The results clearly show how the extracted work decreases with the time spent in the process. Therefore, more work is extracted when more entropy is produced, and the minimum work principle does not apply in this situation. In fact, this is to be expected because both the initial and the final state of the protocol are highly non-passive, and thus the conditions described in Sec.\  \ref{Sec:Optimal protocols for work extraction and minimum work principle} are not satisfied.  However, when using an effective description in terms Gibbs states, we would have predicted that it is always beneficial to use a quasi-static, reversible protocol since condition \eqref{eq:condpositivetemp} is satisfied for the case described in Fig.\ \ref{figLocalQuenchesII}.


\begin{figure}
   \includegraphics[width=1.\columnwidth]{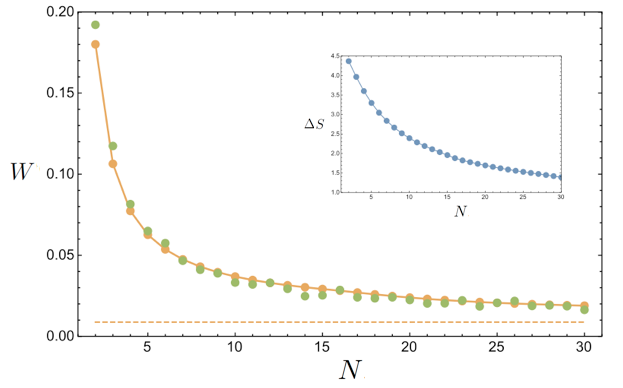}
    \caption{The extracted work achieved with only local transformations on the state of the system. As an initial state we take 
    the one specified by $K=32$, $\tr (a_1^{\dagger} a_1 \rho_S )=0.1$, and $n=150$. For the initial Hamiltonian, we take 
    $\epsilon_0=0.1$, $\epsilon_i=1$ $\forall i\neq 1$, $g=0.5$. As a protocol we perform a first quench to $\epsilon_1=1.6$, followed by $N-1$ equidistant quenches back to the original Hamiltonian. The different points correspond to the exact unitary evolution (in green), to the effective evolution in terms  GGE states (in yellow), and to infinitesimally slow protocol ($N \rightarrow \infty$). As in \ref{figGlobalQuenches}, the real evolution is obtained by letting system and bath interact for a sufficiently long time (chosen at random).
    }
\label{figLocalQuenchesII}
\end{figure}

\section{Conclusions}


In this work, we have brought together the fields of research on equilibration and quantum heat engines. The main contribution of this work is to go beyond the usual paradigm 
of thermodynamics where work is extracted from a system in weak thermal contact with an infinite heat bath at a given fixed temperature. Instead, we consider closed quantum many-body systems of finite size and with strong coupling between its constituents. We make use of recent insights into the study of states out of equilibrium: closed many body systems do not equilibrate, but can be effectively described as if they had equilibrated when looking at a restricted, although most relevant, class of observables. The effective equilibrium state that describes the system for these observables is, however, not necessarily given by a Gibbs state; and even if so, its temperature will not remain constant under repeated quenches. In this case the effective equilibrium state is given by the time averaged state, the GGE or the Gibbs state, depending on the particular kind of system considered, as well as the family of observables that are taken into account.

{With this in mind, we have put forward a framework that studies work extraction of closed many body systems, incorporating Hamiltonian quenches as well as equilibrations 
according to the three models mentioned before. 
We do not only assume that effective equilibrium state is a good description of the state evolving after a single quench, but also that such an equilibrium state can be taken as 
the initial state to describe further evolutions under subsequent quenches. This model, which is successfully tested for the model of free fermions, is what allows us to describe a closed system similarly to the way open systems (in contact with baths) are described in conventional thermodynamics. Thus, we can formulate similar questions regarding work and entropy production and indeed recover many of the phenomena present for open systems.}

In particular, we provide stringent conditions for the absence of entropy production in quasi-static protocols. This turns out to be intimately related to the optimal protocols for work extraction and the minimum work principle, which roughly speaking states that the work performed on the system is minimal for the slowest
realisation of a given process. We find that the minimum-work principle can break down in the presence of a large number of conserved quantities, while it remains intact if system and bath together can be well described by a Gibbs ensemble, even in the strongly interacting regime. This is shown numerically with the paradigmatic example of free fermions for which the extracted work decreases with the time spent in the process if we consider the GGE as equilibration model, but the minimum work principle still applies when the Gibbs description is assumed. It is the hope that the present work stimulates further studies at the intersection of the theory of quantum thermal machines and quantum many-body systems.

\subsection*{Acknowledgements} We acknowledge funding from the A.-v.-H., the BMBF (Q.com), the 
 EU (RAQUEL, SIQS, AQuS), the DFG (EI 519/7-1, CRC 183), the ERC (TAQ) and the Studienstiftung des Deutschen Volkes. M.~P.-L. acknowledges funding from the Spanish MINECO (Severo Ochoa grant SEV-2015-0522) and Grant No.~FPU13/05988. A.~R.\  is supported by the Beatriu de Pinos fellowship (BP-DGR 2013), the Spanish project FOQUS, and the Generalitat de Catalunya (SGR 874). All authors thank the EU COST Action No.~MP1209, \emph{Thermodynamics in the quantum regime}.
 
\subsection*{Note added}
Upon completion of this work, three manuscripts appeared that
 address topics of thermodynamics of quantum systems with multiple conserved quantities 
 \cite{ResourceGGE1,ResourceGGE2,ResourceGGE3}. While there is no actual overlap in content of the 
 present work with that body of work---in which a resource-theoretic mindset is advocated---and 
 the four manuscripts complement each other, the flurry of interest still 
can be seen as a manifestation of the excitement about studying how quantum thermodynamic have to be 
 altered in the situation of a number of conserved quantities being present.


\begin{thebibliography}{10}
\bibitem{footnoteproduct}
\bibinfo{note}{Strictly speaking, the later examples of fermionic systems do
  not have this tensor product structure, as they are defined on a fermionic
  Fock-space. This poses no problems for our considerations, though.}
\end{thebibliography}

\newpage
\appendix
\section{Conserved quantities on the GGE}\label{sec:appGGE}

Here we discuss which are the physical arguments that justify the choice of a given set of conserved quantities that lead to a GGE. This question can be argued from two different approaches. 
On the one hand, one can argue that the relevant conserved quantities are 
the ones that are experimentally accessible and, hence, must be given beforehand.
This was the spirit of the seminal work of Jaynes \cite{Jaynes1957,Jaynes1957a}.
The objection against this approach is that it is \emph{subjective}, in the sense that the set of experimentally accessible observables
depends on the experimentalist.
On the other hand, one could take an \emph{objective} perspective and think that the relevant conserved quantities 
are precisely the ones that make the GGE as close as possible to the diagonal ensemble independently of the capabilities of the
experimentalist \cite{Sels2015}.
Within this approach, the notion of physically relevant is provided by how much an observable is able to reduce the distance between the GGE and the diagonal ensemble 
by being added into the set of conserved quantities that defines the GGE.
More specifically, in \cite{Sels2015} the distance between the time averaged state and the GGE is taken by the Kullback-Leibler (KL) distance (relative entropy) leading to
\begin{equation}
D(\ta(H),\gge(H,\{Q_i\}))=S(\gge(H,\{Q_i\}))-S(\ta(H)),\nonumber
\end{equation}
which is always positive and where we have omitted the initial state $\rho$ for brevity. 

In practice, given an $\varepsilon > 0$, the conserved quantities are successively added to the set of conserved quantities,
until $D(\ta(\rho,H),\gge(\rho,H,\{Q_i\}))\le \varepsilon$. By the Pinsker's inequality, this guarantees 
the physical indistinguishability between the two ensembles, i.e.,
\begin{equation}
\sum_\ell\left| \Tr(B_\ell(\ta-\gge) \right| \le \sqrt{2\varepsilon}\, ,
\end{equation}
for any positive operator valued measure (POVM) $B$.
The addition of operators to the set of conserved quantities is done as follows. Given a set of $j$ conserved quantities,
the new conserved quantity $j+1$ is introduced such that reduces as much as possible the entropy
\begin{equation}
\min_{Q_{j+1}} S(\gge(\rho,H,\{Q_i\}_{i=1}^{j+1})).
\end{equation}
In the subsequent sections we will study what are differences between the thermodynamics
given the Gibbs and the GGE as equilibration models.

\section{Time-average equilibration model - dissipation and reversibility}
\label{sec:reversible_processes}

In this section we show that it is possible to have dissipation, i.e., entropy-production, in an infinitely slow process within the time average equilibration model. 
Let us introduce the following example. 
We consider the Hamiltonians given by
\begin{equation}
H(\lambda_x,\lambda_z)=\lambda_x \sigma_x+\lambda_z \sigma_z
\end{equation}
and the continuous trajectory for $-1\le u \le 1$
\begin{equation}
 \lambda(u)=(\lambda_x(u),\lambda_z(u))=\left\{
 \begin{array}{ll}
 (-u,0) & \textrm{if } \ -1 \le u < 0 \\
 (0,u) & \textrm{if } \ 0\le u < 1 \\
 \end{array}\right.
\end{equation}
starting from an eigenstate of $\sigma_x$. 

For $-1\le u < 0$, the equilibration processes do not do anything to the state 
since the eigenbasis of the Hamiltonian is the eigenbasis of $\sigma_x$ and the system is left in its
initial state with zero entropy. But then, from $u=0$ on, the system is de-phased in the eigenbasis of $\sigma_z$ which is mutually orthogonal to the one of $\sigma_x$ and the entropy suddenly jumps to $\log 2$.
The reason for that is that although the Hamiltonians $H(\epsilon,0)$ and $H(0,\epsilon)$, with $\epsilon > 0$ arbitrarily small, are very close in the Hamiltonian space, their eigenbasis are totally different.

To avoid such effects, it is sufficient that not only the Hamiltonian trajectory is continuous, but also that the eigenvectors can be chosen in a smooth manner, i.e., so that each eigenvector $\ket{E_k(u)}$ is a smooth curve parametrized by $u$. More explicitly, the eigenvalues $p_k(u+\delta u)$ of the density matrix at parameter $u+\delta u$ can be written
in terms of the eigenvalues of the density matrix $\omega(u)$ at time $u$, as
\begin{align}
p_k(u+\delta u)&=\bra{E_k(u+\delta u)}\ta(u)\ket{E_k(u+\delta u)}\nonumber \\
&=\sum_{k'} p_{k'}(u)|\bra{E_{k'}(u)}E_k(u+\delta u)\rangle|^2, 
\end{align}
where we have used that the eigenvalues of $\ta(u+\delta u)$ are simply the diagonal elements of $\ta(u)$ in the basis given by $\ket{E_k(u+\delta u)}$. Let us now assume differentiability of the eigenbasis, i.e.,
\begin{align}
\ket{E_k(u+\delta u)}&=\ket{E_k(u)}+ \ket{X_k(u)}\delta u+ O(\delta u^2),
\end{align}
with $\mathrm{Re}\braket{E_{k'}(u)}{X_k(u)}=0$ due to ortho-normalisation.
Then we get
\begin{align}
p_k(u+\delta u)&= \sum_{k'}p_{k'}(u)\left(\delta_{k'k} + \delta u\, 2\mathrm{Re}\braket{E_{k'}(u)}{X_k(u)}\right) \nonumber\\ &\quad\quad  + \delta u^2 \sum_{k'}p_{k'}(u)  |\braket{E_{k'}(u)}{X_k(u)}|^2\nonumber\\
&= p_k(u) + O(\delta u^2).
\end{align}
This implies that the populations of the density matrix of the system are constant
in the slow process limit $\delta u\to 0$.

A natural way to guarantee that the Hamiltonian eigenbasis changes continuously along
the Hamiltonian trajectory is to restrict ourselves to smooth trajectories, in the sense that the tangent vectors to the curve in the Hamiltonian space are also continuous.

\section{Physically relevant situation of quasi-static processes for the Gibbs ensemble}
\label{sec:examples-constant-entropy-gibbs}

In this Appendix we discuss the entropy production of quasi-static processes with the model of equilibration given by $\omega_{\text{Gibbs}}$. In particular we show Result \ref{res:gibbsentropyproduction} and provide other lemmas that are used in the proof and that are interesting on its own. 

\begin{lemma}[General condition for entropy production within Gibbs model] \label{lemma:gibbsentropyproduction} 
Consider a quasi-static process along a trajectory of Hamiltonians $H(u)$ and an initial state $\rho(0)=e^{-\beta(0) H(0)}/Z$; 
if there exists \emph{any} smooth function $u\mapsto f(u)\neq 0 \, \forall u$ with $f(0)=\beta(0)$ such that 
\begin{equation}\label{eq:entropyconstant_lemma}
S\left( \frac{e^{-f(u) H(u)} }{Z} \right) =S(\rho(0))
\end{equation}
then the quasi-static process along $u\mapsto H(u)$ has no entropy production.
\end{lemma}
\begin{proof}Defining the family of states
\begin{equation}
	\Omega_f(u):=\frac{e^{-f(u) H(u)} }{Z},
\end{equation}	
Lemma \ref{lemma:gibbsentropyproduction} can be shown by noting that eq.\ \eqref{eq:entropyconstant_lemma} implies that
\begin{equation}
\frac{\d S (\omega_f) }{\d u}= f(u) \tr\left(\frac{\d \Omega_f(u)}{\d u}H(u)\right) = 0.
\end{equation}
Taking the equality at the r.h.s., one sees that the state $\omega_f(u)$ fulfils condition \eqref{eq:paralleltransportcontinuum} and hence, $\Omega_f(u)=\omega(u)$ and in turn, $S(\omega(0))=S(\omega(u))$. In other words, any function $f(u)$ that---playing the role of the inverse temperature $\beta(u)$---keeps the entropy constant, will also fulfill the energy conservation condition given by \eqref{eq:paralleltransportcontinuum}, so that $f(u)=\beta(u)$. 
\end{proof}

Lemma \ref{lemma:gibbsentropyproduction} can be used to answer whether there is entropy production 
given a quasi-static process defined by $H(u)$ with $0\le u\le 1$ and initial state $\omega_{\text{Gibbs}}(0)$. We provide now two examples.

\begin{result}\label{example1} Let us refer to the ground state degeneracy of a Hamiltonian $H$ as $d_{\text{g}}(H)$. {Consider an initial and final Hamiltonian $H(0)$ and $H(1)$ such that $d_{\text{g}}(H(0))\geq d_{\text{g}}(H(1))$ and initial state $\omega_{\text{Gibbs}}(0)=e^{-\beta(0)H(0)}/Z(H(0))$ with $\beta(0)>0$. Then, \emph{any} quasi-static trajectory $H(u)$ that satisfies  $d_{\text{g}}(H(0))\geq d_{\text{g}}(H(u)) \geq d_{\text{g}}(H(1))$ for all $u\in[0,1]$ will keep the entropy constant.}
\end{result}

\begin{proof} {First, let us invoke the fact I) that for any Hamiltonian $H$ and any entropy $S\in (\log d_{\text{g}}, \log D)$, there is a finite $\beta_S$ such that the Gibbs state of inverse 
temperature $\beta_S$ has entropy $S$. Now, let us consider the premise given above of a trajectory $u\mapsto H(u)$, so that the ground state degeneracy satisfies $d_{\text{g}}(H(0))\geq d_{\text{g}}(H(u)) \geq d_{\text{g}}(H(1))$ for all $u\in[0,1]$. 
This implies that can choose a function $u\mapsto f(u)$ 
such that $S(\Omega_{f}(u))=S(\omega(0)_{\text{Gibbs}})$ for all $u\in[0,1]$. That this is the case can be seen at $u=0$ just using that $\beta(0)>0$ and hence, the entropy of the initial state is at least $\log (d_g(H(0))$. Hence, it lies within the limits where fact I) applies. For any other $u>0$ we just apply the same reasoning and the premise that $d_{\text{g}}(H(0))\geq d_{\text{g}}(H(u)) \geq d_{\text{g}}(H(1))$ for all $u\in[0,1]$. Since the path of Hamiltonians is smooth, it follows that the function $f$ 
is also smooth.}

Lastly, by Lemma $\ref{res:gibbsentropyproduction}$ this function satisfies 
$f(u)=\beta(u)$, where $\beta(u)$ is the inverse temperature of the quasi-static process. Hence, such a process keeps the entropy constant.
\end{proof}

\begin{result} [Formal version of Resut \ref{res:gibbsentropyproduction} in the main text] \label{example2} Consider an initial and final Hamiltonian $H(0)$ and $H(1)$ and initial state $\omega_{\text{Gibbs}}(0)=e^{-\beta(0) H(0)}/Z$ with finite $\beta(0)>0$. If there exist finite $\beta^{*}>0$ so that $S(\omega_{\text{Gibbs}}(0))=S(e^{-\beta^* H(1)}/Z)$, then any quasi-static trajectory 
$H(u)$ with $d_g(H(u))=1$ for all $u$ in the open interval $u\in (0,1)$, is such there is no entropy production.
\end{result}
\begin{proof} Using the same argument as in the proof of Example \ref{example1}, we find that thermal states of non-degenerate Hamiltonians can take any entropy between $0$ and $\log d$. This implies that we can then find a smooth function $f(u)$ such that $S(\Omega_{f(u)})=S(\omega_{\text{Gibbs}}(0))$ for all $u<1$. But since we assume that a suitable $\beta^*$ exists we can smoothly rescale the Hamiltonians along the trajectory to make sure that $f(1)=\beta^*$, obtaining $S(\Omega_{f(u)})=S(\omega_{\text{Gibbs}}(0)$ for all $u\in[0,1]$.
This ensures by Lemma \ref{lemma:gibbsentropyproduction} that such quasi-static processes exhibit no entropy production.
\end{proof}

\section{Optimal protocols for work extraction with Gibbs ensembles}\label{sec:gibbsworkextraction}

\subsection{The case of unrestricted Hamiltonians}\label{sec:gibbsunrestricted}
First we  consider an idealised scenario where one has full control over the global Hamiltonian $H$. That is, the Hamiltonians $H^{(i)}$ at the $i$-th step of the protocol can be chosen to be any Hamiltonian. Given this maximal level of control, we would like to identify the optimal protocol for work extraction.

We have initial conditions described by a pair of state and Hamiltonian $\rho^{(0)}$ and $H^{(0)}$ respectively. The goal is to extract work by performing a \emph{cyclic} protocol, where $H^{(N)}=H^{(0)}$. Importantly, we will no longer assume that the initial state is in a Gibbs state with respect to the initial Hamiltonian $H^{(0)}$.
After the first quench, the state thermalises to $\omega_{\text{Gibbs}}^{(1)}=\Omega_{\text{Gibbs}}(\rho^{(0)},H^{(1)})$. Hence, from that moment onwards, the minimum work principle can be applied implying that it is optimal to come back to $H^{(0)}$ by a protocol that does not create entropy. 

The only remaining question concerning the optimal protocol is 
to which  Hamiltonian $H^{(1)}$ the first quench is to be performed. This can be straightforwardly answered by expressing the total work, as in  \eqref{eq:workinitialfinal},
\begin{equation} \label{WextGibbs}
W=\tr\left((\rho^{(0)}-\omega_{\text{Gibbs}}^{(N)}
)
H^{(0)}\right),
\end{equation}
where by eq.\ \eqref{eq:i-thgibbs}, we see that $\omega_{\text{Gibbs}}^{(N)}$ is a Gibbs state with inverse temperature $\beta^{(N)}$. Arguing in
the same way as in the minimum work principle, we obtain that the optimal protocol is the one which has no entropy production. Note that a protocol creating zero entropy is only possible for initial states $\rho^{(0)}$ such that $S(\rho^{(0)})=S(e^{-\beta^{(N)} H^{(0)}}/Z)$ for some $\beta^{(N)}>0$, as discussed in Result \ref{res:gibbsentropyproduction}. Here we provide the steps of a protocol that achieves zero entropy production if that condition is met, which is, as discussed above, the protocol that extracts the maximum amount of work. This protocol reads:

\begin{enumerate}
\item Apply first a quench from $H^{(0)}$ to $H^{(1)}=k\ln( \rho^{(0)})$ for any $k \in \mathbb{R}^{-}$.
\item Let the system equilibrate to $\omega_{\text{Gibbs}}^{(1)}:=\gibbs(\rho^{(0)},H^{(1)})$ given by \eqref{eq:i-thgibbs}. The condition of average energy conservation implies that $\beta^{(1)}=-1/k$, and 
thus, $\omega_{\text{Gibbs}}^{(1)}=\rho^{(0)}$.

\item Apply a quasi-static process (a sequence of infinitesimal quenches and equilibrations) from $H^{(1)}$ to $H^{(0)}$. Such process keeps the entropy constant $S(\rho^{(1)})= S(\rho^{(N)})$, as discussed in Sec.\ \ref{sec:entropyproduction}.
\end{enumerate}
This protocol resembles the optimal protocol of work extraction for the model of equilibration of eq.\ \eqref{eq:thermalequilibration} \cite{Esposito2011,Aberg2013a}; however, 
the first quench is chosen to a different Hamiltonian $H^{(1)}$. 


\subsection{Work extraction with restricted Hamiltonians and Gibbs ensembles}
\label{sec:gibbsrestricted}
 We now consider the restricted case where $H^{(i)} \in \mc{H}$ and $\mc{H}$ is a given set of Hamiltonians. While we will later be interested in the case where restriction are
  such that we can only change the initial Hamiltonian \emph{locally} on the subsystem $S$, 
  so that 
  \begin{equation}
  \mc{H}_{\text{Local}}=\{ H \: \: | \: \: H=H^{(0)}+H_{S} \otimes \id_{B} \}, 
  \end{equation}
  we will keep the discussion completely general.

In the same way as in the case of unrestricted Hamiltonians, a maximum amount of work will be extracted if we minimise the final  
energy, as expressed by eq.~\eqref{eq:workinitialfinal}. Since the  final state is by assumption a Gibbs state, it is therefore optimal to end up with a Gibbs state with minimal possible positive temperature (every state with negative temperature has higher  energy than all states with positive temperature). This is clearly possible if the initial state already has an effective positive temperature with respect to any Hamiltonian in $\mc{H}$.  We will assume from now on that this is the case.

Considering steps $1.-3.$ of protocol in
Sec.\ \ref{sec:gibbsunrestricted},  one can easily see that  step $1.$
cannot  be applied  if $k\ln  (\rho^{(0)})\notin \mc{H}$.  Instead, we
will make a quench $H^{(0)} \mapsto H^{(1)}$ with
\begin{equation}
H^{(1)}=\mathrm{argmin}_{H\in\mc{H}} \: \:S\left(\Omega_{\text{Gibbs}(\rho^{(0)},H)}\right),
\end{equation}
while steps $2.-3.$ are not modified by the restrictions on
$\mc{H}$.

\section{Optimal protocol of work extraction for time average equilibration and unrestricted	 Hamiltonians}\label{sec:apptimeaverageopt}

We  now construct an explicit protocol that saturates the bound
\begin{equation}
W \leq \tr(\rho^{(0)} H^{(0)}) - \tr(\omega_{\text{TA}}^{*} H^{(0)})
\end{equation}
in the limit of $N\rightarrow \infty$, where $N$ is the number of quenches performed. Here, $\omega_{\text{TA}}^*$ is a state with the following properties: i) it has the same eigenvalues as $\rho^{(0)}$, ii) it is diagonal in the basis of $H^{(0)}$, iii) it is passive, i.e., its eigenvalues are ordered in non-increasing order with increasing energy. Given the initial state $\rho^{(0)}$, let us denote by $U$ the unitary that diagonalises the initial state, such that $U\rho^{(0)}U^{\dagger}=D$. The first step of the protocol is to make a quench $H^{(0)} \mapsto H^{(1)}$ with $H^{(1)}=U^{\dagger}H^{(0)}U$. Since $\rho^{(0)}$ is diagonal in the eigenbasis of $H^{(1)}$, it follows that the first equilibration process to the time averaged state will not alter the state, that is, $\omega_{\text{TA}}^{(1)}=\rho^{(0)}$. The second step is to perform $N/2$ quenches (followed each by an equilibration process) in a given trajectory from $H^{(1)}$ back to the initial Hamiltonian $H^{(0)}$. Note that in the limit of $N \rightarrow \infty$ this is a quasi-static process, thus the state $\omega_{\text{TA}}^{(N/2)}$ is diagonal with respect to $H^{(0)}$ and with the same eigenvalues as $D$. The next step is to find some unitary $V$ that orders the eigenvalues of $\omega_{\text{TA}}^{N/2}$, in such a way that we have
\begin{equation}
	V\omega_{\text{TA}}^{N/2} V^{ \dagger}=\omega_{\text{TA}}^*:=\sum_k (\rho^{(0)})^{\downarrow}_k P^{(0)}_k,
\end{equation}	
where $({\rho^{(0)}})^\downarrow_k$ denotes the list of eigenvalues of $\rho^{(0)}$ ordered in a non-increasing manner with increasing energy and the $P^{(0)}_k$ are the energy-eigenprojectors of $H^{(0)}$.
As in the previous step, now first perform a quench to $H^{(N/2+1)}=V^{\dagger} H^{(0)} V$ and return to $H^{(N)}=H^{(0)}$ in a quasi-static process, so that in the limit of $N\rightarrow \infty$ we obtain $\omega_{\text{TA}}^{N} = \omega_{\text{TA}}^*$.

\section{Work extraction with time-average equilibration}\label{app:protocol_qs_topassive}

In this section we present the optimal protocol of work extraction between an initial and final Hamiltonian $H(0)$ and $H(1)$ respectively, from an initial state $\omega_{\text{TA}}(0)$. This protocol consists on the quasi-static realisation of the following trajectory $H(u)$: Let us denote the initial Hamiltonian as $H(0)=\sum_i E_i(0) \ketbra{E_i(0)}{E_i(0)}$ and equivalently for the final $H(1)$. Let us assume no degenerate eigenspaces for simplicity (the generalisation to the case with degenerate subspaces is straightforward) so that the initial state is simply given by $\omega_{\text{TA}}(0)=\sum_i p_i \ketbra{E_i(0)}{E_i(0)}$. Then, the quasi-static realisation of the following trajectory of Hamiltonian leaves the final state $\omega_{\text{TA}}^{\text{q.s.}}$ with the same spectrum and passive with respect to $H(1)$:
\begin{enumerate}
\item Change the eigenvalues smoothly from $\{E_i(0)\}_i$ to $\{E_i(u_1)\}_i$ while leaving the eigenstates invariant. Note that in this part of the protocol the state remains also invariant, so that $\omega_{\text{TA}}(u_1)=\omega_{\text{TA}}(0)$. The final eigenvalues $E_i(u_1)$ are chosen so that $\omega_{\text{TA}}(u_1)$ is passive with respect to $H(u_1)=\sum_i E_i(u_1)\ketbra{E_i(0)}{E_i(0)}$ and that the spectrum of $H(u_1)$ coincides with the one of $H(1)$.
\item Given the conditions on the spectrum of $H(u_1)$ and $H(1)$, one can identify $E_j(1)=E_i(u_1)$. In this second part of the protocol we define a smooth trajectory from $u_1$ to $u_2$ where only the eigenvectors are changed as $\ket{E_i(u_1)} \rightarrow \ket{E_i(u_2)}=\ket{E_j(u_2)}$. By definition, after this second step the final Hamiltonian $H(u_2)$ is indeed the desired final Hamiltonian so that $H(u_2)=H(1)$. Also, this second step from $u_1$ to $u_2$ keeps the state passive, so that the final state $\omega_{\text{TA}}(u_2)$ is passive with respect to the desired final Hamiltonian.
\end{enumerate}
However, note that this protocol requires global control over the Hamiltonians. Once we can only control some part of the Hamiltonian, all the available smooth trajectories might lead to a non-passive final state in the quasi-static realisation, so that it can become beneficial to use a protocol with a finite number of quenches which results in entropy-production.

\section{Free fermionic systems}
\subsection{Correlation matrices, time evolution, and entropy}
\label{FreeFermionsI}
We consider Hamiltonians of the type
\begin{equation}
H=\sum_{i,j}c_{i,j} a_i^{\dagger} a_j
\label{initialH}
\end{equation}
where the operators $a_i, a_i^{\dagger}$ satisfy the fermionic anti-commutation relations,
\begin{align}
\{a_i,a_j^{\dagger}\}&=\delta_{i,j} , \\
\{a_i,a_j\}&=\{a_i^{\dagger},a_j^{\dagger}\}=0.
\label{eq:CommRelations}
\end{align}
Since the matrix $c$ in \eqref{initialH} is Hermitian, it can be diagonalised by a unitary operator, $c=A D A^{\dagger}$, where $A A^{\dagger}=1 $ and $D={\rm diag}\{\epsilon_1,\dots,\epsilon_n \}$. The Hamiltonian then can be expressed as,
\begin{equation}
H=\sum_k \epsilon_k \eta_k^{\dagger} \eta_k, 
\end{equation}
with 
\begin{align}
\label{etas}
&\eta_k =\sum_{j}A_{j,k}^{*} a_j ,\\
&\eta_k^{\dagger}=\sum_{j}A_{j,k} a_j^{\dagger}.
\end{align}
The fermionic operators $\eta_k^\dagger,\eta_k$ are usually referred to as \emph{normal modes}.
The unitarity of $A$ ensures that the transformation preserves the commutation relations,
\begin{eqnarray}
\{\eta_k,\eta_l^{\dagger} \}=\sum_{i,j}A_{k,i} A^{*}_{l,j} \{a_i,a_j^{\dagger}\} =\delta_{k,l}.
\end{eqnarray}
where we used \eqref{eq:CommRelations}.

In the following, we will describe states within the framework of correlation matrices.  Define the entries of the correlation matrix $ \gamma(\rho)$ corresponding to $\rho$ as
\begin{equation}
\gamma_a \left( \rho \right)_{i,j}={\rm Tr} ( a_i^{\dagger} a_j \rho ) .
\label{appgammai,j}
\end{equation}
Notice that the diagonal elements represent the occupation probabilities, or populations, of each physical fermion. The correlation matrix in the diagonal basis $\gamma_{\eta} ( \rho)_{i,j}={\rm Tr} ( \eta_i^{\dagger} \eta_j \rho ) $ is related to $\gamma_a$ through  $\gamma_{\eta}=A^{T}\gamma_{a} A^{*}$. The diagonal elements of $\gamma_{\eta}$, corresponding to the populations of the normal modes, play an important role, and we denote them by $p_k$,
\begin{equation}
p_k = {\rm Tr} ( \eta_k^{\dagger} \eta_k \rho ).
\end{equation}
The time evolution of $\gamma(\rho)$ under $H$, $\rho(t)=e^{-\ii Ht}\rho e^{\ii Ht}$, can be easily computed in the Heisenberg picture,
\begin{eqnarray}
\dot{\eta_k}&=&\ii [H,\eta_k]=-\ii E_k \eta_k, \\
\eta_k(t)&=& e^{-\ii E_kt}\eta_k,
\end{eqnarray}
where we have used $\{\eta_i,\eta_j^{\dagger}\}=\delta_{i,j}$ and $\eta_k^2=0$. Therefore, on the one hand,
it follows that
\begin{equation}
\gamma_{\eta} (\rho(t))=e^{\ii tD} \gamma_{\eta} (\rho) e^{-\ii tD}
\end{equation}
with $D={\rm diag}\{E_1,\dots,E_n \}$. In the original basis it reads, 
\begin{equation}
\gamma_a (\rho(t))=U \gamma_a(\rho) U^{\dagger}
\end{equation}
with $U=A^{*}e^{\ii tD}A^{T}$. 
On the other, the time averaged state, which is defined as, 
\begin{equation}
\langle \rho \rangle_t = \lim_{T\rightarrow \infty} \frac{1}{T} \int_0^{T} \rho(t),
\end{equation}
is represented simply by
\begin{equation}
\gamma_{\eta} (\langle \rho \rangle_t) =\langle \gamma_{\eta} (\rho) \rangle_t= \Gamma \left[ \gamma_{\eta} (\rho(t))  \right],
\end{equation}
where $\Gamma$ corresponds to a de-phasing operation. In fact, this correlation matrix is the same as the one of the GGE
where the conserved quantities are the normal modes $\eta_k^{\dagger} \eta_k$, i.e., 
\begin{equation}
\gamma \left(\gge(\rho,H,\{\eta^\dagger_k \eta_k\})\right)= \gamma\left(\average{\rho}_t\right).
\label{eq:GGETimeAverage}
\end{equation}
Note however, that this does not imply that the full quantum state of the GGE is the same as the time-averaged state.

A particularly important class of fermionic states is given by \emph{Gaussian states}. In this situation of fixed particle number, 
such Gaussian states are completely determined from the correlation matrix. In particular, eigenstates and thermal states of free 
fermionic Hamiltonians are Gaussian states, but clearly also the GGEs given above, as they are obtained by maximizing the entropy given the expectation values of the operators $\eta^\dagger_k\eta_k$.

If a state $\rho$ is Gaussian, the entropy of $\rho$ can be calculated as
\begin{equation}
S(\rho)=\sum_k \mathrm{H}(d_k) ,
\end{equation}
where $d_k$ are the eigenvalues of $\gamma(\rho)$, and $\mathrm{H}(p)=-p \ln p -(1-p) \ln (1-p)$. This fact allows us to study entropy-production numerically for large systems.

\subsection{Work extraction for free fermions}
\label{FreeFermionsIII}
Here we find a bound for work extraction protocols, which, as discussed in the main text, is equivalent to finding a lower bound on the final energy, $\tr{(\omega_{\text{GGE}}^{(N)}H^{(0)})}$, with $H^{(0)}=\sum_k \epsilon_k^{(0)} \eta^{\dagger}_k \eta_k^{(0)}$.
From our considerations in Section \ref{FreeFermionsI}, it follows that under the joint operation of a quench, 
\begin{equation}
H^{(i)}=\sum_{i,j} c_{i,j}^{(i)} a^{\dagger}_i a_j \mapsto H^{(i+1)}=\sum_{i,j} c_{i,j}^{(i+1)} a^{\dagger}_i a_j, 
\end{equation}
followed by an equilibration process, our effective description in terms of GGE states takes the form
\begin{equation}
 \gamma_a(\omega_{\text{GGE}}^{(i+1)})=A_{(i+1)}^{*} \Gamma\left[ A_{(i+1)}^{T} \gamma_a(\omega_{\text{GGE}}^{(i)}) A_{(i+1)}^{*}      \right] A_{(i+1)}^{T}
 \label{apptransf}
\end{equation}
where $\Gamma$ is a de-phasing operation, and $c^{(i+1)}=A_{i+1} D A^{\dagger}_{i+1}$, with $D$ a diagonal matrix. Let $\{d_k^{(i+1)}\}$ and $\{d_k^{(i)}\}$ 
be the eigenvalues of $\gamma_a(\omega_{\text{GGE}}^{(i+1)})$ and $\gamma_a(\omega_{\text{GGE}}^{(i)})$, respectively. Under \eqref{apptransf}, they are related through a doubly stochastic matrix,
\begin{equation}
d_k^{(i+1)}=\sum C_{k,l} d_l^{(i)}
\end{equation}
with $\sum_k C_{k,l}=\sum_l C_{k,l}=1$. Therefore, the eigenvalues of the final state $\gamma_a(\omega_{\text{GGE}}^{(N)})$ can also be expressed as a stochastic combination of the eigenvalues of $\gamma_a(\rho^{(0)})$, $\{d_k^{(0)}\}$. It now follows from basic notions of the theory of majorisation that,
\begin{equation}
\tr{(\omega_{\text{GGE}}^{(N)}H^{(0)})} \geq \sum_{k=1}^{n} (d_k^{(0)})^{\downarrow} (\epsilon_k^{(0)})^{\uparrow}=\tr{(\omega_{\text{GGE}}^{*}H^{(0)})} 
\label{appwgge*}
\end{equation}
where $^{\uparrow}$ and $^{\downarrow}$ reflect lists ordered in increasing (decreasing) order. This provides the bound \eqref{pord}. 

We now construct an explicit protocol that achieves this bound in the limit $N\rightarrow \infty$, where $N$ is the number of quenches performed. Let $\gamma_a(\rho^{(0)})$ be the correlation matrix of $\rho^{(0)}$ as in \eqref{appgammai,j}. First, find some $U$ 
 that diagonalises  $\gamma(\rho^{(0)})$, 
 \begin{equation}
 	U\gamma(\rho^{(0)})U^{\dagger}=D, 
\end{equation}
and make a quench to 
 \begin{equation}
H^{(1)}=U^{T} H^{(0)} U^{*}. 
\end{equation}
Since $\gamma(\rho^{(0)})$ is diagonal in the new basis, it follows that $\omega_{\text{GGE}}^{(1)}=\rho^{(0)}$, i.e., the state is not changed during the equilibration process. Now, slowly rotate back to the original Hamiltonian, by performing $N/2$ quenches (followed by equilibration processes) until  $H^{(0)}$ is reached. At the end the state, $\rho^{(N/2)}$ is (approximately) diagonal with respect to the original Hamiltonian,  $H^{(0)}$. Next, find some $V$ that order the populations of $\gamma(\rho^{(N/2)})$, so that $V \gamma(\rho^{(N/2)})V^{\dagger}$ satisfies \eqref{pord}. As before, perform a quench to 
 \begin{equation}
H^{(1)}=V^{T} H^{(0)} V^{*}, 
\end{equation}
and slowly come back to the original Hamiltonian by performing $N/2$ quenches. This process 
give rise to the desired  final state $\omega_{\text{GGE}}^{*}$ in the limit of infinitesimally slow transformations, i.e., in the limit $N \rightarrow \infty$.  The optimal protocol is therefore reversible, and it agrees with our intuition that slow processes are better for work extraction. 
 
Importantly, note that these results for the free fermions are completely analogue to the case of time average equilibrium state, as detailed in Sec.\ \ref{sec:timeaveragework}. Indeed, the optimal final state resembles a passive state, which is the optimal final state for work extracting protocols using time-averaged states. However, it should be stressed that the GGE equilibration model considered for free fermions does not coincide in general with the time averaged state. Indeed, this difference can be highlighted by looking at the final state obtained for the time average model in comparison with the final state of the GGE equilibration for free fermions. In the former, one ends up with a passive state. This implies, for $n$ fermions, $2^n$ energy populations decrease with the energy. On the other hand, for the GGE model of equilibration considered here, the final state $\omega^*_{\text{GGE}}$ is such only the $n$ populations of the normal modes need to be in decreasing order. These two states are in general not the same.


For example, consider a three-fermion system with Hamiltonian 
\begin{equation}
H=\epsilon_1 \eta^{\dagger}_1 \eta_1+\epsilon_2 \eta^{\dagger}_2 \eta_2+\epsilon_3 \eta^{\dagger}_3 \eta_3 
\end{equation}
and a state $\rho$ with $\tr(\eta^{\dagger}_i \eta_i \rho)=p_i$ with $i=1,2,3$. The quantum state $\rho$ and $H$ can be written as
\begin{align}
H=&{\rm diag} \{0,\epsilon_1,\epsilon_2,\epsilon_3,\epsilon_1+\epsilon_2,\epsilon_2+\epsilon_3,\epsilon_1+\epsilon_3,\epsilon_1+\epsilon_2+\epsilon_3 \}
\nonumber\\
\rho=&{\rm diag} \{(1-p_1)(1-p_2)(1-p_3),p_1(1-p_2)(1-p_3),\nonumber\\
&p_2(1-p_1)(1-p_3),p_3(1-p_1)(1-p_2), p_2 p_1(1-p_3),\nonumber\\
&p_2 p_3(1-p_1),p_1 p_3(1-p_2),p_1 p_2 p_3 \}.
\end{align}
If we now choose $\epsilon_1=1$, $\epsilon_2=2$, $\epsilon_3=2.5$; and $p_1=0.4$, $p_2=0.3$, and $p_3=0.1$; we obtain that $\rho$ is not passive but has the form of $\omega_{\text{GGE}}^{*}$.
The origin of the difference is the set of operations in which every state is defined. Passive states arise as optimal states for work extraction protocols if any unitary operation can be performed to the system, or, equivalently, every cyclic process in which the system remains thermally isolated. On the other hand, states $\omega_{\text{GGE}}^{*}$ become optimal when the set of operations corresponds to (arbitrary) quenches to quadratic Hamiltonians, which is in general more constraint that the set of unitary operations.  Within this constrained set of operations, they become optimal.
\clearpage
\end{document}